\documentclass[12pt,a4paper]{article}
\usepackage[T1]{fontenc}
\usepackage[utf8]{inputenc}
\usepackage[a4paper,margin=0.75in]{geometry}
\usepackage{amsmath,amssymb,amsthm,mathtools,bm}
\usepackage{microtype}
\usepackage{graphicx}
\usepackage{booktabs}
\usepackage{xcolor}
\usepackage[colorlinks=true,linkcolor=blue,citecolor=blue,urlcolor=blue]{hyperref}
\usepackage[numbers,sort&compress]{natbib}
\usepackage{authblk}
\usepackage{orcidlink}
\numberwithin{equation}{section}
\allowdisplaybreaks[2]
\graphicspath{{figures/}}
\hypersetup{colorlinks=true,citecolor=blue,linkcolor=blue,urlcolor=blue,
  pdftitle={Order-Sensitive Quantum Distinguishability in Hawking Greybody Scattering},
  pdfauthor={Nazir A. Ganaie},
  pdfsubject={Quantum channel order, information geometry, and Hawking greybody scattering},
  pdfkeywords={quantum channel order, BKM information geometry, Gaussian quantum channels, greybody factors, Hawking radiation, Schwarzschild-de Sitter}}

\newcommand{\Tr}{\operatorname{Tr}}

\newcommand{\id}{\operatorname{id}}
\newcommand{\cD}{\mathcal D}
\newcommand{\cH}{\mathcal H}
\newcommand{\cL}{\mathcal L}
\newcommand{\cS}{\mathcal S}

\newcommand{\cR}{\mathcal R}
\newcommand{\cT}{\mathcal T}

\newcommand{\cC}{\mathcal C}

\newcommand{\sandD}{\widetilde D}
\newcommand{\eps}{\varepsilon}
\newcommand{\norm}[1]{\left\lVert#1\right\rVert}

\newcommand{\dd}{\mathrm d}

\newtheorem{theorem}{Theorem}
\newtheorem{proposition}{Proposition}
\newtheorem{lemma}{Lemma}
\newtheorem{corollary}{Corollary}
\theoremstyle{remark}

\title{Order-Sensitive Quantum Distinguishability in Hawking Greybody Scattering}

\author{Nazir A. Ganaie\,\orcidlink{0009-0004-8627-4105}\thanks{\href{mailto:nazirahmadgan.82225@jk.gov.in}{nazirahmadgan.82225@jk.gov.in}}}
\affil{Department of Physics, National Institute of Technology Srinagar, Srinagar 190006, India}

\date{}

\begin{document}

\maketitle

\begin{abstract}
We study how composition order becomes operationally distinguishable when a phase-sensitive Gaussian control is applied before or after black-hole greybody scattering. The scattering process is formulated on calibrated asymptotic mode ports as a bosonic attenuation channel, and the control is a calibrated squeezing transformation. For finite-dimensional faithful families near a common faithful output, the local symmetrized Umegaki response is governed by the Bogoliubov--Kubo--Mori norm of the channel-composition defect; the corresponding bosonic problem is solved exactly at the Gaussian-state level. The exact attenuation--squeezing response yields an infrared exponent-transfer law relating greybody transmission and occupation scalings to distinguishability. For a massless scalar on Schwarzschild, matched low-frequency scattering gives an explicit fixed-partial-wave logarithmic coefficient and a locally integrable packet response. At finite squeezing, the infrared outputs approach distinct pure Gaussian supports, producing finite sandwiched R\'enyi divergence below unit order, logarithmic Umegaki growth at unit order, and divergence above unit order. Schwarzschild--de Sitter boundary prescriptions generate controlled source-isolated and product-horizon crossover sectors. The results apply to calibrated asymptotic-port compositions and to the reduced finite-band selected-mode realization under the stated narrowband and channel assumptions.
\end{abstract}

\section{Introduction}
\label{sec:intro}

The spectrum detected outside a black hole is shaped by two distinct ingredients: quantum-field excitation at the horizon and frequency-dependent propagation through the exterior curvature potential.  This separation makes Hawking emission a natural setting for asking whether the order of a controlled quantum operation and geometrically fixed scattering can itself be encoded in the outgoing state.  We consider a calibrated phase-sensitive Gaussian control composed in the two definite orders with the Schwarzschild greybody map.  The comparison is defined mode by mode on asymptotic scattering ports, so the question is quantitative: how does the distinguishability of the two ordered outputs scale in the infrared, and which part of that scaling is fixed by the black-hole scattering data?

Relative entropy supplies a natural local measure of this order sensitivity.  Umegaki relative entropy governs asymmetric quantum-state discrimination \cite{Umegaki1962,OgawaNagaoka2000}, while its Hessian at a faithful state is the Bogoliubov--Kubo--Mori (BKM) metric \cite{PetzSudar1996,LesniewskiRuskai1999}.  Sandwiched R\'enyi divergences extend the comparison to a support-sensitive one-parameter family with data processing for R\'enyi order at least one half \cite{FrankLieb2013,Beigi2013,Jencova2018,Jencova2021}.  Channel divergences and multistep process formalisms provide the corresponding language for comparing composed quantum operations \cite{CooneyMosonyiWilde2016,WangWilde2019Channels,ChiribellaDArianoPerinotti2009,PollockEtAl2018}.  In the present setting these established structures are used to isolate the tangent generated by reversing two channel words.  For finite-dimensional faithful families meeting at a common faithful output, we prove that the quadratic symmetrized Umegaki response is the BKM norm of this channel-composition defect.  The black-hole application is then treated independently through exact state-resolved Gaussian formulas on bosonic Fock space.

Gaussian-channel descriptions of black-hole transformations provide the second ingredient.  Bosonic Gaussian channels have been used to classify black-hole communication maps \cite{BradlerAdami2015}, and recent phase-space treatments develop the same language for Hawking and analogue-gravity settings \cite{Brady2025GaussianLectures}.  More recent work has made neighboring elements of the present problem explicit: Schwarzschild greybody transmissivity has been represented as Gaussian attenuation for asymptotic modes \cite{Mondal2026ReflectedEntropy}; scalar Schwarzschild scattering and Hawking emission have been analyzed in an $S$-matrix framework \cite{AlencarPicancoZarro2026}; ordering-sensitive detector states have been studied in a relativistic KMS setting \cite{Rotondo2026}; and symmetrized relative entropy has been related to BKM susceptibility in a modular construction \cite{Chatterjee2026}.  These results leave open the black-hole scattering problem addressed here: a closed order-response law that converts greybody transmission and Hawking occupation into infrared distinguishability, together with its fixed-partial-wave coefficient and its behavior when the ordered states approach different supports.

Selected-input and stabilized order divergences for two completely positive trace-preserving maps first determine the local order geometry on the faithful finite-dimensional state space.  The corresponding bosonic realization is then obtained exactly for a centered one-mode squeezing transformation composed with a phase-insensitive attenuator.  For vacuum attenuation the weak-squeezing susceptibility depends on the transmissivity $\Gamma$ and incident occupation $n$ through the product $\Gamma n$.  If $\Gamma(x)\sim g x^\alpha$ and $n(x)\sim\kappa x^{-\beta}$ as $x\to0^+$, the resulting exponent-transfer theorem gives logarithmic, finite, or power-suppressed order response according to the sign of $\alpha-\beta$, together with a separate finite-transmission branch.  Third, we insert the black-hole scattering and occupation laws into this exact Gaussian response and verify the required coefficients independently.

The principal realization is a massless minimally coupled scalar on four-dimensional Schwarzschild spacetime in the Unruh state, with $G=\hbar=c=k_{\mathrm B}=1$.  Here $M$ is the black-hole mass, $\omega>0$ the Killing frequency, and $\ell\in\mathbb N_0$ the partial-wave index.  The Hawking occupation obeys $n_H(\omega)\sim(8\pi M\omega)^{-1}$, while matched low-frequency Regge--Wheeler scattering gives $\Gamma_\ell(\omega)\sim g_\ell(M\omega)^{2\ell+2}$.  The exponent-transfer theorem therefore yields
\begin{equation}
\chi_\ell(\omega)
=
2(2\ell+1)\ln\frac{1}{M\omega}
+
2\ln\frac{8\pi}{g_\ell}
+o(1),
\qquad
g_\ell=
\left[
\frac{2^{2\ell+2}(\ell !)^{3}}
{(2\ell)!(2\ell+1)!}
\right]^2 ,
\label{eq:introIR}
\end{equation}
for the squeezing-parameter susceptibility at fixed $\ell$ as $M\omega\to0^+$.  The coefficient follows from the zero-frequency Legendre solution and overlap matching, and is checked by flux conservation, horizon-data convergence, an independent areal-radius integration, and an independent asymptotic extraction.  The monochromatic logarithm is locally integrable; finite-band packets remain finite, and the mode-summed analysis supplies a sufficient finiteness condition under the stated angular and squeezing envelopes.

Finite squeezing exposes a second infrared effect that is invisible in the weak-control coefficient alone.  The two ordered Schwarzschild outputs approach distinct pure Gaussian supports.  Consequently, the symmetrized sandwiched R\'enyi divergence has a support-dependent threshold: it approaches a finite limit for $1/2\le q<1$, the Umegaki member at $q=1$ grows logarithmically, and the symmetrized divergence diverges for $q>1$, whereas purified distance, trace distance, and the specified homodyne statistic remain finite.  This separates covariance-driven infrared scaling from the support geometry of the limiting states.

Schwarzschild--de Sitter (SdS) is used as a controlled extension that separates geometric transmission from complementary-port boundary data.  For a nonextremal four-dimensional static patch we compare a source-isolated diagnostic prescription with a stationary nonequilibrium product-horizon thermal benchmark whose marginals carry the two horizon temperatures \cite{BrumJoras2015}.  Known low-frequency $s$-wave scattering gives finite zero-frequency transmission at minimal coupling and quadratic suppression for nonzero curvature coupling \cite{CrispinoHiguchiOliveiraRocha2013,deCesareMirandaPorfyriadis2026}.  The same exponent-transfer law then generates distinct infrared sectors, while joint small-hole/low-frequency scalings resolve the Schwarzschild crossover.  These boundary prescriptions define scattering benchmarks; they do not constitute a unique global equilibrium state.

The operational scope is fixed by calibrated asymptotic ports.  The pre-scattering branch acts on the asymptotic up-input algebra before reciprocal Regge--Wheeler scattering, and the post-scattering branch applies the phase-transported realization of the same calibrated abstract Gaussian operation on the transmitted out algebra.  The scattering phase defines an exact mode-frame identification; a finite-band compensator approximates that identification over a selected packet.  A reduced selected-$s$-wave capture--squeeze--release construction supplies one finite-band realization under lossless, Markovian, dispersion-free, and narrowband input--output assumptions.  A localized four-dimensional apparatus realizing the up-input port requires an additional dynamical model for angular support, mediators, noise, material stress, and backreaction.  The analytical results apply to the calibrated scattering-port compositions and their output-state distinguishability.

The channel-composition defect generates an exact attenuation--squeezing response whose infrared exponent transfer maps black-hole scattering data to the explicit Schwarzschild coefficient in Eq.~\eqref{eq:introIR} and to the support-sensitive R\'enyi hierarchy.  Sections~\ref{sec:definition} and \ref{sec:geometry} define the order diagnostics and their local geometry.  Section~\ref{sec:gaussian} derives the exact Gaussian response.  Section~\ref{sec:schwarzschild} develops the Schwarzschild realization, the controlled SdS extension, packet regularization, phase calibration, and detector-level readout.  Section~\ref{sec:discussion} states the physical scope and limitations.  Appendices~\ref{app:informationGeometryProofs}--\ref{app:globalConsistency} contain the proofs, scattering derivations and independent checks, SdS asymptotics, phase-calibration analysis, selected-mode receiver construction, and global summability bounds.

\section{Ordered quantum interventions}
\label{sec:definition}

Let $\Phi$ and $\Psi$ be CPTP endomorphisms of a finite-dimensional system $A$.  For $\alpha\in[1/2,1)\cup(1,\infty)$ the sandwiched R\'enyi divergence is
\begin{equation}
\sandD_\alpha(\rho\Vert\sigma)
=\frac{1}{\alpha-1}\ln\Tr\!\left[
\left(\sigma^{\frac{1-\alpha}{2\alpha}}\rho\sigma^{\frac{1-\alpha}{2\alpha}}\right)^\alpha
\right],
\label{eq:sandwiched}
\end{equation}
with the standard support convention; its continuous $\alpha=1$ member is
$\sandD_1(\rho\Vert\sigma)=D(\rho\Vert\sigma)=\Tr[\rho(\ln\rho-\ln\sigma)]$.
We use the symmetrized extended-valued divergence
$J_\alpha(\rho,\sigma)=\tfrac12[\sandD_\alpha(\rho\Vert\sigma)+\sandD_\alpha(\sigma\Vert\rho)]$.
The selected-input order divergence is
\begin{equation}
\cH_\alpha(\Phi,\Psi\mid\rho)
=J_\alpha((\Psi\circ\Phi)(\rho),(\Phi\circ\Psi)(\rho)).
\label{eq:holonomy}
\end{equation}
All channel words are ordinary definite compositions.  At $\alpha=1$, the two directed relative entropies
$D((\Psi\Phi)(\rho)\Vert(\Phi\Psi)(\rho))$ and
$D((\Phi\Psi)(\rho)\Vert(\Psi\Phi)(\rho))$
are the primary asymmetric hypothesis-testing exponents.  Their arithmetic mean $\cH_1$ is the symmetric order diagnostic used here.  For an even number $n$ of independently prepared outputs, the balanced revealed-orientation hypotheses
\begin{equation}
H_0^{(n)}=(\rho_A\otimes\rho_B)^{\otimes n/2},\qquad
H_1^{(n)}=(\rho_B\otimes\rho_A)^{\otimes n/2}
\label{eq:symSteinProtocol}
\end{equation}
have optimal type-II exponent per output copy $\cH_1$ by additivity and quantum Stein's lemma.  Thus the standard symmetrized Umegaki divergence used throughout the article also has a balanced revealed-orientation Stein interpretation for this specific protocol.

To obtain a channel-level quantity, stabilize the composition-generated channel box with a reference $R\simeq A$,
\begin{equation}
\widehat{\cH}_\alpha(\Phi,\Psi)
=\sup_{\rho_{RA}}J_\alpha\!\Big(
[\id_R\!\otimes(\Psi\Phi)](\rho_{RA}),
[\id_R\!\otimes(\Phi\Psi)](\rho_{RA})\Big).
\label{eq:stabilized}
\end{equation}
The reference dimension $d_R=d_A$ is sufficient and the optimization may be restricted to pure probes.

For a bounded global certificate, let
$f(\rho,\sigma)=\|\sqrt\rho\sqrt\sigma\|_1$ and $P(\rho,\sigma)=\sqrt{1-f(\rho,\sigma)^2}$,
and define

\begin{equation}
\widehat P_{\rm ord}(\Phi,\Psi)
=\sup_{\rho_{RA}}P\!\Big(
[\id_R\!\otimes(\Psi\Phi)](\rho_{RA}),
[\id_R\!\otimes(\Phi\Psi)](\rho_{RA})\Big).
\label{eq:purifiedOrder}
\end{equation}

\begin{proposition}[Global order certificates]
\label{prop:bounded-order-distance}
The stabilized quantities satisfy
\begin{equation}
\widehat P_{\rm ord}=0\Longleftrightarrow\Psi\Phi=\Phi\Psi,
\label{eq:purifiedFaithfulness}
\end{equation}
\begin{equation}
\|\Psi\Phi-\Phi\Psi\|_\diamond\le2\widehat P_{\rm ord}\le2,
\label{eq:purifiedDiamond}
\end{equation}
and
\begin{equation}
\widehat P_{\rm ord}=\sqrt{1-e^{-\widehat{\cH}_{1/2}}}.
\label{eq:purifiedRenyiRelation}
\end{equation}
Moreover, for every $\alpha\ge1/2$,
\begin{equation}
\widehat{\cH}_\alpha=0\Longleftrightarrow\Psi\Phi=\Phi\Psi,
\label{eq:completefaithful}
\end{equation}
and common CPTP postprocessing $\Lambda$ contracts the stabilized divergence,
\begin{equation}
\widehat{\cH}_\alpha(\Lambda\Psi\Phi,\Lambda\Phi\Psi)
\le\widehat{\cH}_\alpha(\Phi,\Psi).
\label{eq:stabilized-postprocess}
\end{equation}
At $\alpha=1$,
\begin{equation}
\|\Psi\Phi-\Phi\Psi\|_\diamond
\le\sqrt{2\widehat{\cH}_1}.
\label{eq:diamond}
\end{equation}
\end{proposition}
The state-level Pinsker specialization is
\begin{equation}
\| (\Psi\Phi)(\rho)-(\Phi\Psi)(\rho)\|_1
\le\sqrt{2\cH_1(\Phi,\Psi\mid\rho)}.
\label{eq:pinsker}
\end{equation}
\begin{corollary}[Output-distance bound]
\label{cor:pinsker}
Equation~\eqref{eq:pinsker} holds for every input state for which $\cH_1$ is finite.
\end{corollary}
\begin{proposition}[Reference-dimension reduction]
\label{prop:reference-reduction}
The supremum in Eq.~\eqref{eq:stabilized} is unchanged by allowing arbitrary finite-dimensional references and may be restricted to pure states with $d_R=d_A$.
\end{proposition}
The finite-dimensional proofs and inherited state-level calculus are collected in Appendix~\ref{app:informationGeometryProofs}.

\section{Information geometry and channel-order tangent susceptibility}
\label{sec:geometry}

Let
$\mathcal N=\Psi\Phi$, $\mathcal M=\Phi\Psi$, and $\mathcal N(\sigma)=\mathcal M(\sigma)=:\tau$
for faithful finite-dimensional states $\sigma$ and $\tau$.  For a faithful $\omega$, define
$\cT_\omega(Y)=\int_0^\infty(\omega+tI)^{-1}Y(\omega+tI)^{-1}\,\dd t$
and the BKM form
$\langle X,Y\rangle_{\omega,{\rm BKM}}=\Tr[X\cT_\omega(Y)]$ and $\|X\|_{\omega,{\rm BKM}}^2=\langle X,X\rangle_{\omega,{\rm BKM}}$.

\begin{lemma}[Two-base-point BKM expansion]
\label{lem:twobase}
For fixed traceless Hermitian tangents $X,Y$ and sufficiently small real $\epsilon$,
\begin{equation}
D(\omega+\epsilon X\Vert\omega+\epsilon Y)
=\frac{\epsilon^2}{2}\|X-Y\|_{\omega,{\rm BKM}}^2+O(|\epsilon|^3),
\label{eq:twobase}
\end{equation}
with a uniform cubic remainder on a faithful neighborhood.
\end{lemma}

Define the composition defect on the input tangent space by
$\Omega_{\Phi,\Psi}^{\sigma\to\tau}:=\left.(\Psi\Phi-\Phi\Psi)\right|_{\mathsf T_\sigma}:\mathsf T_\sigma\to\mathsf T_\tau$.

\begin{theorem}[Composition-defect BKM pullback]
\label{thm:curvature}
For $\rho_\epsilon=\sigma+\epsilon X$ within the faithful state space,
\begin{equation}
\cH_1(\Phi,\Psi\mid\rho_\epsilon)
=\frac{\epsilon^2}{2}\|\Omega_{\Phi,\Psi}^{\sigma\to\tau}X\|_{\tau,{\rm BKM}}^2
+O(|\epsilon|^3),
\label{eq:curvaturelaw}
\end{equation}
so the directional susceptibility is
\begin{equation}
\chi_{\sigma\to\tau}^{\rm ord}(X;\Phi,\Psi)
:=\lim_{\epsilon\to0}\frac{2\cH_1}{\epsilon^2}
=\|\Omega_{\Phi,\Psi}^{\sigma\to\tau}X\|_{\tau,{\rm BKM}}^2.
\label{eq:directionalsusceptibility}
\end{equation}
\end{theorem}
\begin{proof}
Linearity gives the two output curves $\tau+\epsilon\mathcal N(X)$ and $\tau+\epsilon\mathcal M(X)$.  Lemma~\ref{lem:twobase} applied in both directions gives the same quadratic coefficient, whose tangent difference is $\Omega_{\Phi,\Psi}^{\sigma\to\tau}X$.
\end{proof}

The common-output condition also gives
$\Omega_{\Phi,\Psi}^{\sigma\to\tau}=0\Longleftrightarrow\Psi\Phi=\Phi\Psi$.
For a smooth model $\rho_\theta$ with $X_i=\partial_i\rho_\theta|_0$, the pullback tensor
$G^{\rm ord}_{ij}=\langle\Omega X_i,\Omega X_j\rangle_{\tau,{\rm BKM}}$
is positive semidefinite, transforms covariantly under reparameterization, and contracts under common downstream CPTP processing.  Optimization over unit BKM tangents is the Rayleigh--Ritz problem for $\Omega^\sharp\Omega$; linear resource constraints restrict that problem to the corresponding BKM subspace.  These geometric extensions, including the optimal and constrained formulas, are proved in Appendix~\ref{app:informationGeometryProofs}.  For $\sigma$-preserving semigroups $\Phi_s=e^{s\mathcal L_1}$ and $\Psi_t=e^{t\mathcal L_2}$,
\begin{equation}
\Omega_{\Phi_s,\Psi_t}X
=st[\mathcal L_2,\mathcal L_1]X
+O(s^2t+st^2),
\label{eq:semigroupDefect}
\end{equation}
and therefore
\begin{equation}
\chi_{\sigma\to\sigma}^{\rm ord}(X;\Phi_s,\Psi_t)
=s^2t^2\|[\mathcal L_2,\mathcal L_1]X\|_{\sigma,{\rm BKM}}^2
+O(s^3t^2+s^2t^3).
\label{eq:semigroupSusceptibility}
\end{equation}
The corresponding state divergence for $\rho_\epsilon=\sigma+\epsilon X$ carries the additional prefactor $\epsilon^2/2$ from Eq.~\eqref{eq:curvaturelaw}.

\section{Exact Gaussian order distinguishability for a Hawking mode}
\label{sec:gaussian}

The stabilized channel-level statements of Sec.~\ref{sec:definition} are finite-dimensional structural results.  From this section onward, the black-hole application is formulated through state-resolved divergences of trace-class single-mode Gaussian states on bosonic Fock space.  Unrestricted infinite-dimensional diamond-norm analysis lies outside this application.  A mode-resolved realization of intervention order is obtained by composing a centered phase-sensitive Gaussian unitary with phase-insensitive attenuation.  The calculation concerns one selected bosonic mode; Sec.~\ref{sec:schwarzschild} supplies a finite-bandwidth realization, while Appendix~\ref{app:informationGeometryProofs} records the conditional operator-algebraic formulation.  The symbol $n_H$ anticipates the Schwarzschild specialization of Sec.~\ref{sec:schwarzschild}; here it denotes an arbitrary thermal occupation.  All logarithms are natural, entropic quantities are measured in nats, all first moments vanish, and
$R=(q,p)^{\mathsf T}$, $[q,p]=i$, and $V_{jk}=\frac12\langle\{R_j,R_k\}\rangle$,
so the vacuum covariance is $I_2/2$.  The dimensionless parameter domain is $n_H,n_E\geq0$, $0\leq\eta\leq1$, and $r\in\mathbb R$.

The thermal input $\tau_{n_H}$ has covariance
\begin{equation}
V_{n_H}=\nu_H I_2,
\qquad
\nu_H=n_H+\frac12.
\label{eq:thermalcov}
\end{equation}
For real squeezing amplitude $r$, the zero-displacement squeezing channel is
$S_r=\operatorname{diag}(e^r,e^{-r})$, with $\cS_r:V\longmapsto S_rVS_r^{\mathsf T}$.
The centered thermal attenuator with transmissivity $\eta$ and environmental occupation $n_E$ acts as
\begin{equation}
\cL_{\eta,n_E}:V\longmapsto
\eta V+(1-\eta)\nu_E I_2,
\qquad
\nu_E=n_E+\frac12.
\label{eq:loss}
\end{equation}
The vacuum Schwarzschild specialization used in Sec.~\ref{sec:schwarzschild} sets $n_E=0$ and identifies $\eta$ with the Regge-Wheeler transmission probability.  Before that specialization, Eqs.~\eqref{eq:thermalcov}-\eqref{eq:loss} define a generic centered one-mode Gaussian model.

The two definite orders are
$\rho_A=(\cL_{\eta,n_E}\circ\cS_r)(\tau_{n_H})$ and $\rho_B=(\cS_r\circ\cL_{\eta,n_E})(\tau_{n_H})$.
The first order applies the mode-selective phase-sensitive operation before attenuation; the second applies the same calibrated abstract Gaussian operation to the attenuated output.  Their distinguishability diagnoses placement relative to the lossy channel for two definite temporal orders.  Coherently controlled order is a process-level extension.

Set
\begin{equation}
\begin{aligned}
a&=\eta\nu_H,\qquad c=(1-\eta)\nu_E,\\
m&=a+c=\frac12+\eta n_H+(1-\eta)n_E.
\end{aligned}
\label{eq:acm}
\end{equation}
Direct composition gives
$V_A=\operatorname{diag}(ae^{2r}+c,ae^{-2r}+c)$ and $V_B=m\,\operatorname{diag}(e^{2r},e^{-2r})$.
Their exact covariance closure defect is
\begin{equation}
V_A-V_B
=c\,\operatorname{diag}(1-e^{2r},1-e^{-2r}).
\label{eq:GaussianCovarianceDefect}
\end{equation}
The order dependence is therefore generated entirely by the covariance entering through the environmental port: attenuation adds this isotropic term after squeezing in $\rho_A$, whereas the same term is squeezed in $\rho_B$.  Even at $n_E=0$, the vacuum contribution $\nu_E=1/2$ produces a nonzero defect whenever $\eta<1$ and $r\neq0$.  Since $\nu_E>0$, Eq.~\eqref{eq:GaussianCovarianceDefect} gives
\begin{equation}
V_A=V_B\Longleftrightarrow r=0\ \text{or}\ \eta=1.
\label{eq:GaussianEqualityCondition}
\end{equation}

Every centered diagonal one-mode covariance has the form $V=\nu\operatorname{diag}(e^{2s},e^{-2s})$, where $\nu=\sqrt{\det V}\geq1/2$ is its symplectic eigenvalue. Hence
\begin{align}
\nu_A&=\sqrt{m^2+4ac\sinh^2r},
\label{eq:nuA}\\
s_A&=\frac14\ln\!\frac{ae^{2r}+c}{ae^{-2r}+c},
\label{eq:sA}\\
\nu_B&=m,\qquad s_B=r.
\label{eq:nuB}
\end{align}
The corresponding thermal occupations are $n_A=\nu_A-1/2$ and $n_B=\nu_B-1/2=m-1/2$.

For $n>0$, define
$g(n)=(n+1)\ln(n+1)-n\ln n$ and $b(n)=\ln(1+1/n)$.
Let $\rho(n,s)$ denote the centered squeezed thermal state with covariance $(n+1/2)\operatorname{diag}(e^{2s},e^{-2s})$.  Specializing the general Gaussian relative-entropy formula to the convention $[q,p]=i$ gives, for faithful states $n_a,n_b>0$ \cite{Parthasarathy2022},
\begin{equation}
\begin{aligned}
D[\rho(n_a,s_a)\Vert\rho(n_b,s_b)]
={}&g(n_b)-g(n_a)\\
&+b(n_b)\nu_a\cosh\!\bigl(2(s_a-s_b)\bigr)\\
&-b(n_b)\nu_b.
\end{aligned}
\label{eq:gaussianD}
\end{equation}
where $\nu_j=n_j+1/2$.  At $s_a=s_b$ this reduces to the thermal relative entropy, and it vanishes for $n_a=n_b$ and $s_a=s_b$, fixing the normalization used below.  No displacement term appears because all states are centered.

\begin{theorem}[Exact squeezing-attenuation order divergence]
\label{thm:gaussianexact}
Assume $m>1/2$, equivalently $\eta n_H+(1-\eta)n_E>0$.  Then both ordered outputs are faithful and their Umegaki order divergence is
\begin{equation}
\begin{aligned}
\cH_{1}^{\rm G}(n_H,\eta,r,n_E)
={}&\frac12\Bigl\{
 b(n_B)\left[\nu_A C-\nu_B\right]\\
&\quad+b(n_A)\left[\nu_B C-\nu_A\right]
\Bigr\},\\
C={}&\cosh\!\left[2(s_A-r)\right].
\end{aligned}
\label{eq:exactGaussianHolonomy}
\end{equation}
It is finite and nonnegative, with
\begin{equation}
\cH_{1}^{\rm G}=0
\quad\Longleftrightarrow\quad
r=0\ \text{or}\ \eta=1.
\label{eq:GaussianHolonomyZeroSet}
\end{equation}
\end{theorem}

\emph{Proof.} See Appendix~\ref{app:centralProofs}.

The excluded surface $m=1/2$ is support singular.  Because all occupations are nonnegative, it is characterized by $\eta n_H=(1-\eta)n_E=0$.  When $r=0$ or $\eta=1$, the two outputs coincide and the order divergence is zero, including at pure outputs.  For $m=1/2$, $r\neq0$, and $\eta<1$, two cases remain.  If $\eta=0$, $\rho_A$ and $\rho_B$ are distinct pure Gaussian states; if $0<\eta<1$, $\rho_B$ is pure whereas $\rho_A$ is mixed because
$\nu_A^2=\frac14+\eta(1-\eta)\sinh^2r>\frac14$.
At least one directed relative entropy is therefore infinite, and
\begin{equation}
\cH_1^{\rm G}=+\infty.
\label{eq:GaussianSupportSingularity}
\end{equation}
This extended-value boundary lies outside the faithful BKM regime of Sec.~\ref{sec:geometry}.

The exact expression also yields a controlled expansion in the intervention strength.  This is a channel-parameter susceptibility at fixed input state, distinct from the state-tangent susceptibility in Eq.~\eqref{eq:directionalsusceptibility}.

\begin{theorem}[Weak-squeezing order susceptibility]
\label{thm:susceptibility}
For fixed $n_H,n_E,\eta$ with $m>1/2$, the order divergence is an even analytic function of $r$ near the origin and
\begin{equation}
\cH_{1}^{\rm G}
=2b\!\left(m-\frac12\right)\frac{c^2}{m}\,r^2
+O(r^4).
\label{eq:weakr}
\end{equation}
The squeezing-strength susceptibility is
\begin{equation}
\chi_{r}^{\rm ord}
:=\left.\frac{\partial^2\cH_1^{\rm G}}{\partial r^2}\right|_{r=0}
=4b\!\left(m-\frac12\right)\frac{c^2}{m}.
\label{eq:chiord}
\end{equation}
The implied $O(r^4)$ constant depends on the fixed parameters and is finite on compact subsets of the faithful domain $m>1/2$.
\end{theorem}

\emph{Proof.} See Appendix~\ref{app:centralProofs}.

Equation~\eqref{eq:chiord} isolates the mechanism already visible in Eq.~\eqref{eq:GaussianCovarianceDefect}: the leading order response is quadratic in the environmental covariance $c$.  The input occupation enters through $m$ and the thermal weight $b(m-1/2)$, which assign an information-geometric cost to the same covariance displacement.  At $\eta=1$, $c=0$ and the two channels commute exactly.  As $m\downarrow1/2$, $b(m-1/2)$ diverges; the full susceptibility diverges along approaches for which $c$ remains nonzero, while it can remain finite or vanish when $c\to0$ sufficiently rapidly.  This path dependence is consistent with the discontinuous support structure summarized in Eq.~\eqref{eq:GaussianSupportSingularity}.

The Gaussian relative-entropy reduction and numerical consistency checks are summarized in Appendices~\ref{app:scatteringFoundations} and~\ref{app:centralProofs}.

The Schwarzschild calculation in Sec.~\ref{sec:schwarzschild} supplies the mode dependence through $n_E=0$, $n_H=n_H(\omega)$, and $\eta=\Gamma_\ell(\omega)$, where $\Gamma_\ell$ is obtained directly from the Regge-Wheeler equation.  The Gaussian identities are algebraic in $n_H$, $n_E$, $r$, and $\eta$; Sec.~\ref{sec:schwarzschild} supplies $\eta$ from direct scattering data.

\section{Black-hole scattering realizations}
\label{sec:schwarzschild}

We specialize Sec.~\ref{sec:gaussian} to a massless minimally coupled scalar on the four-dimensional Schwarzschild exterior, using signature $(-,+,+,+)$ and $G=c=\hbar=k_{\mathrm B}=1$.  Here $M>0$, $r>2M$, $\omega>0$, $\ell\in\mathbb N_0$, and
$\phi=e^{-i\omega t}Y_{\ell m}(\theta,\varphi)u_{\ell\omega}(r)/r$, the radial function satisfies
\begin{align}
\frac{\dd^2u_{\ell\omega}}{\dd r_*^2}
+\left[\omega^2-V_\ell(r)\right]u_{\ell\omega}&=0,
\label{eq:RW}\\
r_*&=r+2M\ln\!\left(\frac{r}{2M}-1\right),
\label{eq:tortoise}
\end{align}
with
$V_\ell(r)=\left(1-\frac{2M}{r}\right)\left[\frac{\ell(\ell+1)}{r^2}+\frac{2M}{r^3}\right]$.
The physical scattering solution is purely ingoing at the future horizon and has the asymptotic form
\begin{align}
u_{\ell\omega}&\sim e^{-i\omega r_*}, && r_*\to-\infty,
\label{eq:RWhorizonBC}\\
u_{\ell\omega}&\sim A_{\ell}^{\rm in}e^{-i\omega r_*}
+A_{\ell}^{\rm out}e^{i\omega r_*}, && r_*\to+\infty.
\label{eq:RWinfinityBC}
\end{align}
The transmitted horizon amplitude is normalized to unity.  The asymptotic coefficients in Eq.~\eqref{eq:RWinfinityBC} may be estimated at large finite radius with the leading plane-wave basis through
\begin{equation}
\begin{aligned}
A_{\ell}^{\rm in}
&=\frac12 e^{i\omega r_*}\left(u-\frac{u'}{i\omega}\right),\\
A_{\ell}^{\rm out}
&=\frac12 e^{-i\omega r_*}\left(u+\frac{u'}{i\omega}\right).
\end{aligned}
\label{eq:RWasymptoticAmplitudes}
\end{equation}
The production extraction instead uses the second-order Jost basis
\begin{equation}
\psi_{\pm}=e^{\pm i\omega r_*}\left(1+\frac{a_{1,\pm}}r+\frac{a_{2,\pm}}{r^2}\right),
\quad
a_{1,\pm}=\frac{\pm i\ell(\ell+1)}{2\omega},
\label{eq:asymptoticJostBasis}
\end{equation}
with $a_{2,\pm}=\ell(\ell+1)[2-\ell(\ell+1)]/(8\omega^2)\pm iM/(2\omega)$.  Equation~\eqref{eq:RWasymptoticAmplitudes} is retained as a lower-order audit, while Eq.~\eqref{eq:asymptoticJostBasis} includes the finite-radius $O(r^{-1})$ and $O(r^{-2})$ corrections in the reported spectrum.  For the conserved radial current $J=(u^*u'-uu'^*)/(2i)$, Eqs.~\eqref{eq:RWhorizonBC} and \eqref{eq:RWinfinityBC} give $J=-\omega$ at the horizon and $J=\omega(|A_{\ell}^{\rm out}|^2-|A_{\ell}^{\rm in}|^2)$ at infinity.  Hence
\begin{equation}
\begin{aligned}
|A_{\ell}^{\rm in}|^2-|A_{\ell}^{\rm out}|^2&=1,\\
\Gamma_\ell(\omega)&=\frac{1}{|A_{\ell}^{\rm in}|^2},
&
\mathcal R_\ell(\omega)&=\left|\frac{A_{\ell}^{\rm out}}{A_{\ell}^{\rm in}}\right|^2,
\end{aligned}
\label{eq:RWgamma}
\end{equation}
so $\Gamma_\ell+\mathcal R_\ell=1$.  The dimensionless transmission probability $\Gamma_\ell$ is independent of the magnetic quantum number $m$.  Because $V_\ell$ is real and time-reversal invariant, transmission is reciprocal: the same $\Gamma_\ell$ governs propagation of the horizon-originating up mode to future null infinity.

The production table uses DOP853 integration at $150$ logarithmic frequencies $10^{-3}\le x=M\omega\le1.6$ for $0\le\ell\le12$, with second-order horizon Frobenius data and the Jost basis in Eq.~\eqref{eq:asymptoticJostBasis}.  The largest flux residual is $3.53\times10^{-10}$; domain/tolerance, Frobenius-order, horizon-offset, independent-coordinate, and asymptotic-basis checks are summarized in Appendix~\ref{app:scatteringFoundations} and Table~\ref{tab:numericalValidation}.

The scalar absorption cross section is $\sigma_{\rm abs}(\omega)=\frac{\pi}{\omega^2}\sum_{\ell}(2\ell+1)\Gamma_\ell(\omega)$.  The plotted $\ell\le12$ sum changes by $4.0\times10^{-16}$ at $x=1.6$ when $13\le\ell\le18$ are added.  At $x=10^{-3}$, $\Gamma_0/(16x^2)=\sigma_{\rm abs}/(16\pi M^2)=1.00643$, approaching $\Gamma_0=16x^2[1+o(1)]$ and the horizon-area limit \cite{UnruhAbsorption1976,DasGibbonsMathur1997,Higuchi2001}; the upper-frequency data begin the known photon-sphere oscillations about $27\pi M^2$ \cite{Sanchez1978,DecaniniEtAl2011,OuldElHadj2025}.  Figure~\ref{fig:rwscattering} displays these transmission and absorption checks.

\begin{figure*}[t]
\centering
\includegraphics[width=0.48\textwidth]{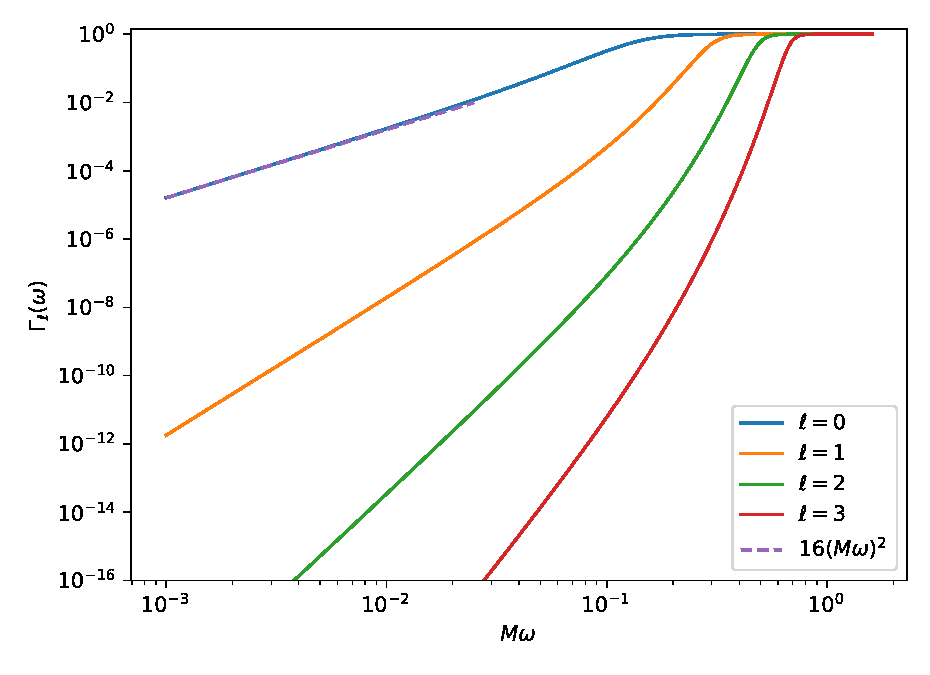}\hfill
\includegraphics[width=0.48\textwidth]{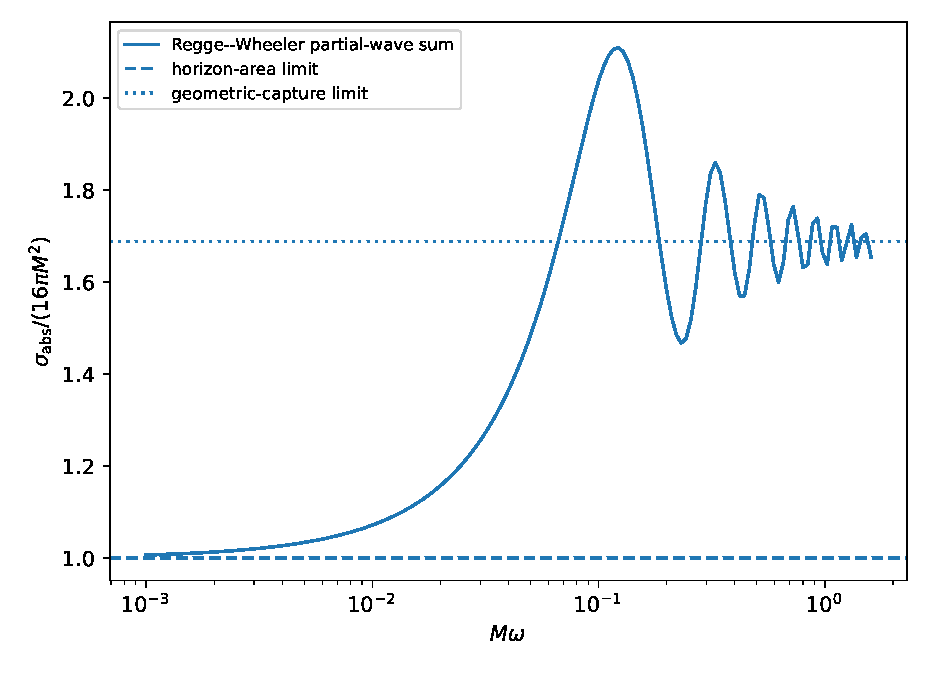}
\caption{Direct scalar Regge-Wheeler scattering on $10^{-3}\leq M\omega\leq1.6$.  Left: transmission probabilities for $\ell=0,1,2,3$ (solid curves) and the infrared $s$-wave law $16(M\omega)^2$ (dashed) over $M\omega<0.025$.  Right: $\sigma_{\rm abs}/(16\pi M^2)$ from the partial-wave sum through $\ell=12$ (solid), with the horizon-area limit $1$ (dashed) and geometric-capture limit $27/16$ (dotted).  The low-frequency area law and the onset of photon-sphere oscillations provide independent normalization checks.}
\label{fig:rwscattering}
\end{figure*}

For the Unruh state the up-mode occupation is
\begin{equation}
n_H(\omega)=\frac{1}{e^{8\pi M\omega}-1},
\label{eq:Hawkingoccupation}
\end{equation}
at $T_H=(8\pi M)^{-1}$ \cite{Hawking1975}.  With the past-null-infinity input in vacuum, the two-port Regge--Wheeler scatterer induces the vacuum attenuator with $\eta=\Gamma_\ell(\omega)$.  Thus $\mathfrak h_\ell(\omega;r)=\cH_1^{\rm G}[n_H(\omega),\Gamma_\ell(\omega),r,0]$ and $\chi_\ell=\left.\partial_r^2\mathfrak h_\ell\right|_{r=0}$.  These state-resolved quantities are $m$ independent; detector response is treated separately.  Reported spectra use the integrated $\Gamma_\ell$ directly.

\begin{lemma}[Scalar Schwarzschild low-frequency coefficient]
\label{lem:gellSchwarzschild}
For every fixed $\ell\in\mathbb N_0$, the massless minimally coupled scalar transmission probability obeys
\begin{equation}
\Gamma_\ell(\omega)
=g_\ell (M\omega)^{2\ell+2}[1+o(1)],
\qquad M\omega\to0^+,
\label{eq:gell}
\end{equation}
with the explicit Schwarzschild coefficient
\begin{equation}
g_\ell
=\left[
\frac{2^{2\ell+2}(\ell!)^3}{(2\ell)!(2\ell+1)!}
\right]^2.
\label{eq:gellExplicit}
\end{equation}
In particular $g_0=16$, $g_1=16/9$, and $g_2=64/2025$.
\end{lemma}

\emph{Proof.} See Appendix~\ref{app:centralProofs}.

In the remainder of this section, $\chi$ denotes the squeezing-parameter susceptibility of the exact Gaussian realization in Eq.~\eqref{eq:chiord}, distinct from the state-tangent BKM susceptibility of Sec.~\ref{sec:geometry}.

\begin{theorem}[Vacuum-attenuation greybody-occupation exponent transfer]
\label{thm:genericExponentLaw}
Let $x\to0^+$ and consider vacuum attenuation with
\begin{equation}
\Gamma(x)=g x^\alpha[1+o(1)],
\qquad
n(x)=\kappa x^{-\beta}[1+o(1)],
\label{eq:genericExponentInput}
\end{equation}
where $g,\kappa>0$, $\alpha\ge0$, and $\beta\ge0$.  For $\alpha>0$, put $d=\alpha-\beta$.  Then
\begin{equation}
\chi(x)=
\begin{cases}
2d\ln(1/x)+2\ln[1/(g\kappa)]+o(1),&d>0,\\[3pt]
\displaystyle \frac{\ln(1+1/(g\kappa))}{g\kappa+1/2}+o(1),&d=0,\\[8pt]
\displaystyle \frac{x^{-2d}}{(g\kappa)^2}[1+o(1)],&d<0.
\end{cases}
\label{eq:genericExponentChi}
\end{equation}
For the finite-transmission case $\alpha=0$, assume $0<g<1$.  If $\beta>0$,
\begin{equation}
\chi(x)=\frac{(1-g)^2}{(g\kappa)^2}x^{2\beta}[1+o(1)],
\label{eq:genericFiniteTransmissionGrowingOccupation}
\end{equation}
while for $\beta=0$,
\begin{equation}
\chi(x)\longrightarrow
\frac{(1-g)^2}{g\kappa+1/2}
\ln\!\left(1+\frac1{g\kappa}\right).
\label{eq:genericFiniteTransmissionFiniteOccupation}
\end{equation}
Thus the $\alpha>0$ sector has logarithmic, finite, and power-suppressed branches according to the sign of $d$, and finite zero-frequency transmission combined with a growing occupation gives the additional power-suppressed law in Eq.~\eqref{eq:genericFiniteTransmissionGrowingOccupation}.  In the logarithmic regime $\alpha>0$, $d>0$, and $\beta>0$, every fixed $r\in\mathbb R$ obeys
\begin{equation}
\begin{split}
\mathfrak h(x;r)
={}&\sinh^2r\left[
 d\ln\frac1x
+\ln\frac1{g\kappa}
-\frac12\ln\cosh(2r)
\right]+o(1),
\end{split}
\label{eq:genericExponentFinite}
\end{equation}
with the remainder uniform for $r$ in compact intervals.
\end{theorem}

\emph{Proof.} See Appendix~\ref{app:centralProofs}.

\begin{theorem}[Partial-wave infrared hierarchy]
\label{thm:partialwaveIR}
Let $x=M\omega$ and let $g_\ell$ be the explicit Schwarzschild coefficient in Eq.~\eqref{eq:gellExplicit}.  Define
\begin{equation}
L_\ell(x):=(2\ell+1)\ln\frac1x+\ln\frac{8\pi}{g_\ell}.
\label{eq:infraredLogarithm}
\end{equation}
Then the vacuum-attenuation susceptibility satisfies
\begin{equation}
\chi_\ell(\omega)=2L_\ell(x)+o(1).
\label{eq:chiellIR}
\end{equation}
For every fixed $r\in\mathbb R$, the finite-amplitude order divergence has the infrared asymptotic
\begin{equation}
\begin{split}
\mathfrak h_\ell(\omega;r)
={}&\sinh^2r
\left[L_\ell(x)-\frac12\ln\!\cosh(2r)\right]+o(1),\\
&x\to0^+.
\end{split}
\label{eq:hellIR}
\end{equation}
where the remainder is uniform for $r$ in every compact interval.  Consequently,
\begin{equation}
\lim_{x\to0^+}
\left[
\lim_{r\to0}\frac{\mathfrak h_\ell(\omega;r)}{r^2}
-L_\ell(x)
\right]=0.
\label{eq:sequentialIRlimit}
\end{equation}
For the scalar $s$ wave, $g_0=16$, and hence
\begin{equation}
\lim_{r\to0}\frac{\mathfrak h_0(\omega;r)}{r^2}
=\frac{\chi_0(\omega)}{2}
=\ln\frac{\pi}{2M\omega}+o(1),
\qquad M\omega\to0^+.
\label{eq:IRlimit}
\end{equation}
If $|r(x)|\leq R$ on $(0,\delta)$ for some finite $R$, then $\int_0^\delta\mathfrak h_\ell(x/M;r(x))\,\dd x<\infty$ for every fixed $\ell$.
\end{theorem}

\emph{Proof.} See Appendix~\ref{app:centralProofs}.

The logarithm in Eq.~\eqref{eq:hellIR} has a specific information-theoretic origin: the two ordered Gaussian states purify toward different one-mode vacua as $x\to0^+$.  Bounded one-shot distances therefore saturate, while Umegaki relative entropy acquires a support-boundary divergence.  The following theorem makes this metric dependence explicit and supplies a R\'enyi threshold across the support boundary.

\begin{theorem}[Support-sensitive R\'enyi threshold and bounded-distance saturation]
\label{thm:renyiInfraredThreshold}
Fix $\ell\in\mathbb N_0$ and a nonzero finite squeezing amplitude $r$.  Let $\rho_{A,\ell}(x;r)$ and $\rho_{B,\ell}(x;r)$ be the two ordered Schwarzschild Gaussian outputs.  Then
\begin{equation}
\rho_{A,\ell}(x;r)\longrightarrow |0\rangle\langle0|,
\qquad
\rho_{B,\ell}(x;r)\longrightarrow S(r)|0\rangle\langle0|S^\dagger(r)
\label{eq:pureInfraredOrderedStates}
\end{equation}
in trace norm as $x=M\omega\to0^+$.  Writing $q$ for the R\'enyi index in this theorem, for $1/2\le q<1$,
\begin{equation}
\lim_{x\to0^+}\mathfrak h_{\ell,q}(x;r)
=\frac{q}{1-q}\ln\cosh r,
\label{eq:renyiInfraredFinite}
\end{equation}
where $\mathfrak h_{\ell,q}$ denotes the symmetrized sandwiched-R\'enyi order divergence at R\'enyi index $q$.  At $q=1$,
\begin{equation}
\mathfrak h_{\ell,1}(x;r)
=\sinh^2r\left[L_\ell(x)-\tfrac12\ln\cosh(2r)\right]+o(1)\to+\infty,
\label{eq:umegakiInfraredBoundary}
\end{equation}
and for every $q>1$,
\begin{equation}
\lim_{x\to0^+}\mathfrak h_{\ell,q}(x;r)=+\infty.
\label{eq:renyiInfraredDivergent}
\end{equation}
The bounded distances instead obey
\begin{align}
P(\rho_{A,\ell},\rho_{B,\ell})
&\longrightarrow\sqrt{1-\operatorname{sech}r},\\
\frac12\|\rho_{A,\ell}-\rho_{B,\ell}\|_1
&\longrightarrow\sqrt{1-\operatorname{sech}r}.
\label{eq:boundedInfraredDistances}
\end{align}
For ideal $q$- or $p$-homodyne readout, the symmetrized classical Kullback-Leibler divergence satisfies
\begin{equation}
J_{\rm hom}^{(q)}\longrightarrow\sinh^2r,
\qquad
J_{\rm hom}^{(p)}\longrightarrow\sinh^2r.
\label{eq:homodyneInfraredSaturation}
\end{equation}
The $q>1$ statement is a support-sensitive divergence of the sandwiched R\'enyi family as the limiting supports separate.  Accessible one-shot distinguishability is represented here by the bounded distances and the specified per-mode homodyne record, which remain finite; the directed Umegaki exponents generate the $q=1$ logarithm.
\end{theorem}

\emph{Proof.} See Appendix~\ref{app:centralProofs}.

The exponent transfer also distinguishes static-patch infrared states.  Consider four-dimensional nonextremal Schwarzschild-de Sitter,
$f(r)=1-\frac{2M}{r}-\frac{\Lambda r^2}{3}$, with $r_h<r<r_c$,
with
\begin{align}
a&:=\frac{r_h}{r_c},\qquad \varepsilon:=\omega r_h,\nonumber\\
\vartheta_h&:=T_hr_h
=\frac{(1-a)(1+2a)}{4\pi(1+a+a^2)},\nonumber\\
\vartheta_c&:=T_cr_h
=\frac{a(1-a)(2+a)}{4\pi(1+a+a^2)}.
\label{eq:sdsTemperatures}
\end{align}
All frequencies and horizon temperatures in this subsection use the same normalization of the static Killing generator associated with the displayed coordinate $t$.  Under a common constant rescaling of that generator, $\omega$, $T_h$, and $T_c$ rescale together; the occupation ratios $\varepsilon/\vartheta_{h,c}$ and the infrared branch exponents are unchanged.  Numerical coefficients below refer to the normalization of the displayed static metric.
The two incoming static-patch ports may be assigned independent centered thermal occupations
\begin{equation}
n_h=(e^{\varepsilon/\vartheta_h}-1)^{-1},
\qquad
n_c=(e^{\varepsilon/\vartheta_c}-1)^{-1}.
\label{eq:sdsTwoOccupations}
\end{equation}
Equation~\eqref{eq:sdsTwoOccupations} defines a stationary nonequilibrium quasifree product-horizon benchmark whose thermal marginals mirror the factorized horizon KMS data of Ref.~\cite{BrumJoras2015}; $n_c=0$ defines the source-isolated diagnostic prescription.  These mode-resolved boundary data do not specify global Hartle--Hawking--Israel, Unruh, or single-temperature KMS states.  Geometry fixes $\eta(\omega)$, boundary data fix $n_h,n_c$, and the Gaussian map fixes the order susceptibility.  If the transmissivity from the event-horizon port to the selected exterior output is $\eta$, Eq.~\eqref{eq:chiord} applies with
$N=\eta n_h+(1-\eta)n_c$ and $c=(1-\eta)(n_c+1/2)$,
and hence
\begin{equation}
\chi^{\rm SdS}=4\ln\!\left(1+\frac1N\right)\frac{c^2}{N+1/2}.
\label{eq:sdsTwoPortExactChi}
\end{equation}

\begin{proposition}[SdS infrared scaling for source-isolated and product-horizon thermal inputs]
\label{prop:sdsTwoHorizonScaling}
For a minimally coupled massless scalar $s$ wave,
\begin{equation}
\eta(\varepsilon)\longrightarrow g_{\rm dS}:=
\frac{4a^2}{(1+a^2)^2}\in(0,1)
\qquad(\varepsilon\to0^+),
\label{eq:sdsGammaConstant}
\end{equation}
a zero-frequency limit valid throughout the nonextremal range \cite{deCesareMirandaPorfyriadis2026,CrispinoHiguchiOliveiraRocha2013}.  For the source-isolated incoming-port prescription, $n_c=0$,
\begin{equation}
\chi_{0,\mathrm{src}}^{\rm SdS}(\varepsilon)
=\frac{(1-g_{\rm dS})^2}{g_{\rm dS}^2\vartheta_h^2}
\varepsilon^2[1+o(1)] .
\label{eq:sdsSuppressedChi}
\end{equation}
For the product-horizon thermal benchmark, with both incoming marginals populated according to Eq.~\eqref{eq:sdsTwoOccupations}, define
$K_a:=g_{\rm dS}\vartheta_h+(1-g_{\rm dS})\vartheta_c$.
Then
\begin{equation}
\chi_{0,\mathrm{2T}}^{\rm SdS}(\varepsilon)
\longrightarrow
4\left[\frac{(1-g_{\rm dS})\vartheta_c}{K_a}\right]^2 .
\label{eq:sdsTwoThermalLimit}
\end{equation}
For nonzero curvature coupling with
$\eta(\varepsilon)=g_\xi\varepsilon^2[1+o(1)]$, $g_\xi>0$ \cite{CrispinoHiguchiOliveiraRocha2013}, the source-isolated prescription gives
\begin{equation}
\chi_{0,\xi,\mathrm{src}}^{\rm SdS}
=2\ln\frac1\varepsilon+O(1),
\label{eq:sdsNonminimalLog}
\end{equation}
while the product-horizon thermal benchmark gives
\begin{equation}
\chi_{0,\xi,\mathrm{2T}}^{\rm SdS}\longrightarrow4.
\label{eq:sdsNonminimalTwoThermal}
\end{equation}
\end{proposition}

\emph{Proof.} See Appendix~\ref{app:sdsProofs}.  The same branch classification holds for $n_j=\kappa_j\varepsilon^{-1}[1+o(1)]$ with finite positive $\kappa_j$, with coefficient-dependent finite limits.  Fixed-$a$ statements are pointwise for $0<a<1$ and uniform on compact nonextremal sets $a\in[\delta,1-\delta]$; Nariai scaling is separate.

The minimally coupled problem also resolves the singular recovery of Schwarzschild.  In the joint small-hole/low-frequency regime, Ref.~\cite{deCesareMirandaPorfyriadis2026} gives
\begin{equation}
\Gamma_{0,\mathrm{MA}}^{\rm SdS}(a,\varepsilon)
=\frac{4(a^2+\varepsilon^2)}{a^2\varepsilon^2+(1+a^2+\varepsilon^2)^2},
\label{eq:sdsMatchedJoint}
\end{equation}
which is symmetric under $a\leftrightarrow\varepsilon$.  Two double scalings expose the source and complementary-port crossovers.  For the source-isolated prescription take $a\to0$ and $\varepsilon=\lambda a^2$; then
\begin{equation}
\chi_{0,\mathrm{src}}^{\rm SdS}\longrightarrow
F(\lambda):=
\frac{\ln[1+\lambda/(4\vartheta_0)]}
{4\vartheta_0/\lambda+1/2},
\qquad \vartheta_0=\frac1{4\pi},
\label{eq:sdsDoubleScalingBranches}
\end{equation}
whereas the product-horizon benchmark with $\varepsilon=\lambda a$ gives
\begin{equation}
\chi_{0,\mathrm{2T}}^{\rm SdS}\longrightarrow
4\pi\lambda\coth(\pi\lambda).
\label{eq:sdsTwoThermalDoubleScaling}
\end{equation}
Equation~\eqref{eq:sdsDoubleScalingBranches} lies between the fixed-$a$ power-suppressed sector and the Schwarzschild logarithmic sector; Eq.~\eqref{eq:sdsTwoThermalDoubleScaling} tends to $4$ as $\lambda\to0$.  Appendix~\ref{app:sdsProofs} gives the complete four-sector source-isolated expansion, the nonuniform large-$\lambda$ product-thermal limit, and the direct radial check.

\begin{figure}[t]
\includegraphics[width=0.92\columnwidth]{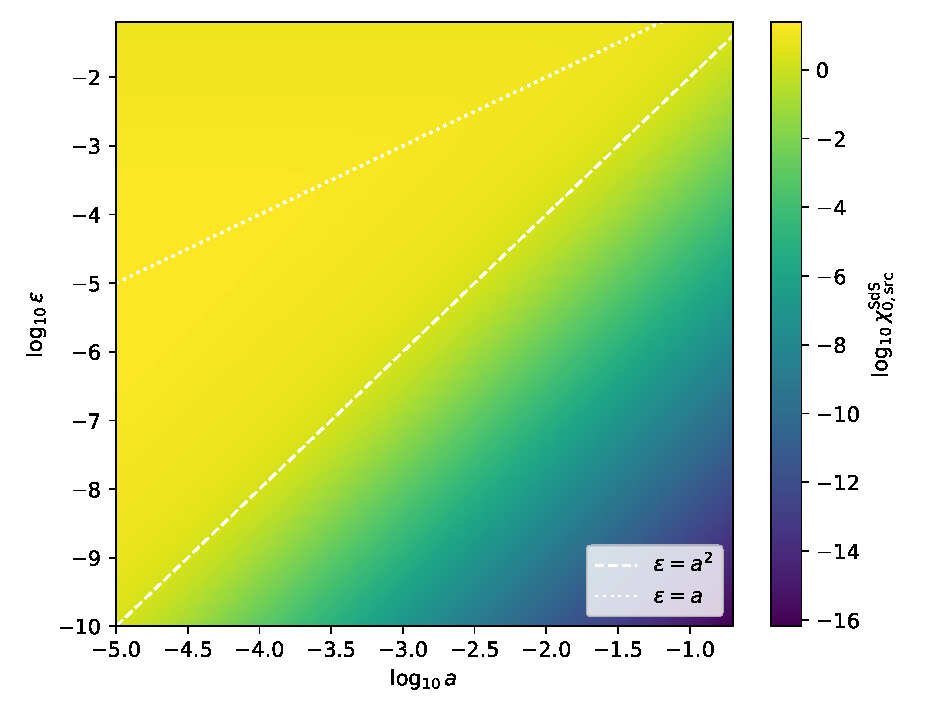}
\caption{Source-isolated minimally coupled SdS order susceptibility from the joint matched transmission in Eq.~\eqref{eq:sdsMatchedJoint}.  Dashed and dotted lines mark $\varepsilon=a^2$ and $\varepsilon=a$, respectively.  The lower sector is power suppressed, the intermediate sectors resolve the double-scaling crossover, and the upper-left sector approaches the Schwarzschild logarithmic law.}
\label{fig:sdsDoubleScaling}
\end{figure}

Direct static-patch integration over $a=0.01,0.05,0.1,0.5,0.9$ reproduces the minimally coupled zero slope, the $\xi=1/6$ slope $2$, and the limit $4a^2/(1+a^2)^2$, with largest normalized flux residual $1.93\times10^{-7}$.  The $a=0.9$ points serve as near-Nariai diagnostics; Appendix~\ref{app:sdsProofs} and Table~\ref{tab:numericalValidation} give the solver and checks.

The explicit coefficient also fixes the large-$\ell$ scale: Stirling's formula gives $\ln(1/g_\ell)=2\ell\ln\ell+O(\ell)$.  Hence a mode-summed protocol requires angular decay of its squeezing envelope.  Appendix~\ref{app:globalConsistency} proves the finite-amplitude global estimate
\begin{equation}
\mathfrak h_\ell\le K_R r_\ell(x)^2
\left[1+\ln_+\frac1{\Gamma_\ell(x)n_H(x)}\right],
\label{eq:nonlinearAngularBound}
\end{equation}
and, with the matched greybody law and a Schwarzschild transmission bound, yields a sufficient angular-frequency summability condition.  Finite angular passbands with square-integrable squeezing profile and the stated logarithmic infrared weight are therefore finite; Appendix~\ref{app:globalConsistency} gives the proof.

The monochromatic ordering is defined on the asymptotic in/out Jost-mode algebras of the Regge--Wheeler problem.  For fixed $(\ell,m)$,
\begin{equation}
b_{\omega}^{\rm out}
=t_\ell(\omega)a_{\omega}^{\rm up}
+r_\ell(\omega)c_{\omega}^{\rm in},
\qquad
t_\ell=\sqrt{\Gamma_\ell}e^{i\phi_\ell},
\label{eq:packetScattering}
\end{equation}
with $c_{\omega}^{\rm in}$ in vacuum.  ``Pre-scattering'' means composition on the asymptotic up-input algebra before the reciprocal two-port scattering unitary; ``post-scattering'' means the phase-transported realization of the same calibrated abstract control on the transmitted out algebra.  This is a scattering-port ordering experiment; a barrier-localized actuator is a separate spacetime interaction model.

The transmission phase defines the exact calibrated mode-frame rotation $\cR_{\phi}$, so
$\cC_{\Gamma,\phi}=\cR_{\phi}\circ\cL_{\Gamma,0}$,
and the same calibrated abstract squeezing operation is transported to the output frame by
$\cS_{\rm out}^{(\phi)}:=\cR_\phi\circ\cS_{\rm in}\circ\cR_\phi^{-1}$.
Phase covariance of the vacuum attenuator gives
\begin{align}
\cC_{\Gamma,\phi}\circ\cS_{\rm in}
&=\cR_\phi\circ(\cL_{\Gamma,0}\circ\cS_{\rm in}),\nonumber\\
\cS_{\rm out}^{(\phi)}\circ\cC_{\Gamma,\phi}
&=\cR_\phi\circ(\cS_{\rm in}\circ\cL_{\Gamma,0}).
\label{eq:phaseCalibratedOrder}
\end{align}
Hence every common-unitary-invariant divergence used here equals the $\Gamma$-only attenuation-squeezing expression in the calibrated mode frame, frequency by frequency.  A residual post-control quadrature-axis mismatch $\delta\phi$ has the exact single-mode penalty derived in Appendix~\ref{app:phaseCalibration},
\begin{align}
0\le\cH_1^{\rm G}(\delta\phi)-\cH_1^{\rm G}(0)
&=K_\phi\sin^2\delta\phi,\nonumber\\
&\le K_\phi(\delta\phi)^2,\qquad K_\phi\ge0.
\label{eq:phaseMismatchBound}
\end{align}
For the Schwarzschild partial-wave infrared limit, let $x=M\omega$ and $y_\ell=\Gamma_\ell n_H$.  The same exact coefficient has the asymptotic form
\begin{align}
K_\phi
&=2y_\ell\sinh^2(2r)\ln\frac1{y_\ell}+O_r(y_\ell),\nonumber\\
&=O_r\!\left[x^{2\ell+1}\ln\frac1x\right]\longrightarrow0.
\label{eq:phaseMismatchIR}
\end{align}
Hence every fixed residual axis error has a vanishing Schwarzschild-infrared penalty.  Appendix~\ref{app:phaseCalibration} gives the proof and direct check.  For an implemented finite-band phase compensator restricted to one adjustable phase and one delay, define
\begin{equation}
\epsilon_{\phi,\rm aff}^2
=\int_B\dd\omega\,|f(\omega)|^2
\left|e^{i[\phi_\ell(\omega)-\phi_\ell^{\rm aff}(\omega)]}-1\right|^2,
\label{eq:phaseAffineError}
\end{equation}
which is invariant under adjustable affine phase $a+b\omega$.  The production packet gives $\epsilon_{\phi,\rm aff}=2.60\times10^{-5}$; Table~\ref{tab:numericalValidation} gives the independent checks.

A selected narrowband one-sided memory realizes capture--squeeze--release in the lossless, Markovian, dispersion-free reduced-$s$-wave port \cite{KiilerichMolmer2019,NurdinJamesYamamoto2016}.  Appendix~\ref{app:receiverProofs} gives the bounded finite-time construction; for $r=0.30$ its packet Umegaki shift is $1.99\times10^{-4}$ nats.  Loss, dispersion, microscopic angular selection, material support, and mediator noise remain outside this receiver model.

For a normalized packet envelope $f$ supported in $B$, define
\begin{align}
A_f&=\int_B\dd\omega\,\overline{f(\omega)}a_{\omega}^{\rm up},
&
1&=\int_B\dd\omega\,|f(\omega)|^2,
\label{eq:packetMode}\\
\tau_f&=\int_B\dd\omega\,|f(\omega)|^2\sqrt{\Gamma_\ell(\omega)},
\nonumber\\
\epsilon_t^2&=\int_B\dd\omega\,|f(\omega)|^2
|\sqrt{\Gamma_\ell(\omega)}-\tau_f|^2,
\label{eq:packetLeakage}\\
\bar n_f&=\int_B\dd\omega\,|f(\omega)|^2n_H(\omega),
\nonumber\\
\epsilon_n^2&=\int_B\dd\omega\,|f(\omega)|^2
|n_H(\omega)-\bar n_f|^2.
\label{eq:packetThermalSpread}
\end{align}
With packet frequency variance $\sigma_\omega^2$, the mean-value theorem gives
\begin{equation}
\epsilon_t\le\sigma_\omega\sup_B|\partial_\omega\sqrt{\Gamma_\ell}|,
\qquad
\epsilon_n\le\sigma_\omega\sup_B|\partial_\omega n_H|.
\label{eq:packetErrorBound}
\end{equation}
The exact transmitted covariance obeys
\begin{equation}
\begin{aligned}
\max_{j,k}|(V_f-V_{\rm flat})_{jk}|
&\le n_B\epsilon_t^2
 +2|\tau_f|e^{|r|}\epsilon_t\epsilon_n,\\
n_B&=\sup_B n_H.
\end{aligned}
\label{eq:packetCovarianceBound}
\end{equation}
Appendix~\ref{app:receiverProofs} proves the covariance bound and gives the finite-memory data.  For packet-matched homodyne variances $v_A,v_B$, $J_{\rm hom}=\frac14(v_A/v_B+v_B/v_A-2)\le\cH_1^{\rm G}$; Gaussian loss and mode mismatch contract the divergence.  At $r=0.30$, $\eta_D=0.8$, overlap $\mu=0.9$, and vacuum rejection, the representative $q$-homodyne value is $4.676\times10^{-2}$ nats.  The appendix and reproducibility files contain the detector construction and scans.

Restoring SI dimensions only for the local semiclassical comparison, define the geometrized mass length
$\mu_M:=G M_{\rm phys}/c^2$ and $\ell_{\rm P}^2=G\hbar/c^3$.
The strongest reduced-field backreaction consistency certificate is
\begin{equation}
8\pi\frac{G}{c^4}\mu_M^2
\sup_{\mathcal R_{\rm ctrl},\,t\ge t_{\rm on}}
|\Delta\langle T_{\hat t\hat t}\rangle|
\le4.545\left(\frac{\ell_{\rm P}}{\mu_M}\right)^2.
\label{eq:localStressBudget}
\end{equation}
Here the natural-unit $M$ equals $\mu_M$.  Equation~\eqref{eq:localStressBudget} is the reduced-field local certificate; Appendix~\ref{app:globalConsistency} gives its derivation, with shell-average bookkeeping retained in the reproducibility files.

\section{Discussion and limitations}
\label{sec:discussion}

Two definite compositions of the same calibrated phase-sensitive control with the same scattering map produce the order-sensitive outputs considered here.  Standard divergence theory supplies positivity, data processing, stabilization, and the BKM Hessian.  Here that Hessian is pulled back along the channel-composition defect and realized exactly for attenuation and squeezing, so the greybody data acquire a calculable information-theoretic image.  This result is adjacent to Gaussian black-hole channel descriptions \cite{BradlerAdami2015,Brady2025GaussianLectures}, the recent greybody-attenuation preprint \cite{Mondal2026ReflectedEntropy}, Schwarzschild $S$-matrix scattering \cite{AlencarPicancoZarro2026}, the recent KMS detector-ordering preprint \cite{Rotondo2026}, and the recent modular BKM-susceptibility preprint \cite{Chatterjee2026}.  For attenuation and squeezing, the order response is closed analytically and yields the exponent-transfer law, the explicit Schwarzschild coefficient, and the support-sensitive R\'enyi hierarchy.

For Schwarzschild, $n_H\sim\omega^{-1}$ and $\Gamma_\ell\sim\omega^{2\ell+2}$ imply $\Gamma_\ell n_H\sim\omega^{2\ell+1}$; the Gaussian susceptibility converts this scale into the logarithm of Theorem~\ref{thm:partialwaveIR}, with the matched coefficient fixing its constant term.  Regular finite-bandwidth packets remain finite, and Appendix~\ref{app:globalConsistency} controls the mode sum.

At finite squeezing the ordered Schwarzschild outputs approach different pure Gaussian supports.  Directed Umegaki exponents and their symmetric mean acquire the logarithmic growth, sandwiched R\'enyi divergences below unit order remain finite, and orders above unity diverge.  Purified distance, trace distance, and the homodyne statistic remain finite, identifying support sensitivity as the source of the entropic infrared enhancement.

SdS separates geometry from complementary-port boundary data.  The source-isolated diagnostic and the product-horizon nonequilibrium thermal benchmark of Ref.~\cite{BrumJoras2015} give different infrared images of the same scattering law: minimal coupling yields a power-suppressed source-isolated response and finite product-thermal limit; nonzero curvature coupling yields a logarithmic source-isolated branch and finite product-thermal limit.  Joint small-hole/low-frequency limits resolve the Schwarzschild crossover.

The operational interpretation is restricted to calibrated asymptotic ports: pre-control is an input--output construction on the up-input algebra before reciprocal scattering, and post-control is the phase-transported realization of the same calibrated abstract Gaussian operation on the transmitted out algebra.  Appendix~\ref{app:receiverProofs} gives a bounded capture--squeeze--release realization within a lossless Markovian reduced chiral port.  A localized exterior apparatus realizing the up-input port introduces apparatus-specific angular selection, support, mediator dynamics, noise, and stress-energy.  Equation~\eqref{eq:localStressBudget} is a reduced-field fixed-background certificate within the selected-port model.

For normal states, Araki relative entropy supplies the operator-algebraic order functional; the BKM formula extends under the regular embedded-family assumptions of Appendix~\ref{app:informationGeometryProofs}.  An intrinsic general type-III channel-order geometry requires additional domain and regularity control.

Natural extensions include Kerr amplification, interacting-field departures from Gaussian closure, nonstationary evaporation, and microscopic controllers with dynamical stress-energy and noise.  Each changes an explicit assumption of the present framework.

\section{Conclusions}
\label{sec:conclusion}

Channel composition order and black-hole scattering data are quantitatively related through the order susceptibility.  For finite-dimensional faithful families meeting at a common faithful output, the local symmetrized Umegaki response is the BKM norm of the channel-composition defect.  For the centered one-mode Gaussian realization, attenuation and squeezing are solved exactly, and the weak-squeezing susceptibility yields an exponent-transfer law that maps the infrared scalings of transmissivity and occupation into logarithmic, finite, or power-suppressed order distinguishability.  For greybody propagation, the exact Gaussian response maps transmission and occupation data into an information-theoretic infrared scaling law, extending the Gaussian-channel description of the propagation map \cite{BradlerAdami2015,Mondal2026ReflectedEntropy}.

For a massless minimally coupled scalar on four-dimensional Schwarzschild spacetime in the Unruh state, the matched low-frequency greybody coefficient and Hawking occupation give, at fixed partial wave,
\begin{equation}
\chi_\ell(\omega)
=
2(2\ell+1)\ln\!\frac{1}{M\omega}
+
2\ln\!\frac{8\pi}{g_\ell}
+o(1),
\qquad
g_\ell=
\left[
\frac{2^{2\ell+2}(\ell!)^3}{(2\ell)!(2\ell+1)!}
\right]^2 ,
\end{equation}
as $M\omega\to0^+$.  The coefficient follows analytically from the zero-frequency Legendre solution and overlap matching; independent scattering integrations and flux tests verify its numerical normalization.  The monochromatic logarithm is locally integrable, and finite-band packets remain finite under the stated squeezing and angular envelopes.  At fixed nonzero squeezing, the two ordered outputs approach distinct pure Gaussian supports.  The resulting support geometry gives a finite symmetrized sandwiched R\'enyi limit for $1/2\le q<1$, logarithmic Umegaki growth at $q=1$, and divergence for $q>1$, while purified distance, trace distance, and the specified homodyne statistic remain finite.

The Schwarzschild--de Sitter extension shows that the infrared order response also depends on the incoming static-patch boundary prescription.  The source-isolated diagnostic and product-horizon thermal benchmark occupy different sectors of the same exponent-transfer law, while the controlled double scalings recover the Schwarzschild crossover using the corresponding low-frequency SdS transmission behavior \cite{CrispinoHiguchiOliveiraRocha2013,deCesareMirandaPorfyriadis2026}.  The results apply to calibrated asymptotic scattering ports.  The reduced selected-mode receiver provides one finite-band realization under lossless, Markovian, dispersion-free, and narrowband input--output assumptions; a localized four-dimensional controller requires additional apparatus dynamics.

Kerr superradiance, interacting departures from Gaussian closure, and nonstationary evaporation require modified scattering, channel, or stationarity assumptions.  Within the stated regime, greybody transmission and occupation scaling determine the infrared conversion of composition order into distinguishability, while the limiting support geometry determines which information measures remain finite or diverge.

\appendix

\section{Information-theoretic and geometric proofs}
\label{app:informationGeometryProofs}

This appendix proves the finite-dimensional structural, stabilization, and BKM-geometry statements used in Secs.~\ref{sec:definition} and \ref{sec:geometry}, and records the conditional operator-algebraic scope statement.

\subsection{Bounded purified-distance order metric}
For the root fidelity $f(\rho,\sigma)=\|\sqrt\rho\sqrt\sigma\|_1$, the sandwiched R\'enyi divergence at $\alpha=1/2$ is
\begin{equation}
\widetilde D_{1/2}(\rho\Vert\sigma)=-2\ln f(\rho,\sigma).
\end{equation}
Since $f$ is symmetric, the state-resolved order divergence obeys
\begin{equation}
P(\rho_{12},\rho_{21})
=\sqrt{1-\exp[-\mathcal H_{1/2}(\Phi,\Psi\mid\rho)]}.
\end{equation}
The function on the right is monotone increasing, so taking the stabilized supremum gives Eq.~\eqref{eq:purifiedRenyiRelation}.  Fidelity is monotone increasing under CPTP maps, hence purified distance contracts under common postprocessing.  It vanishes precisely when its two state arguments coincide, so the same product-probe spanning argument as in Proposition~\ref{prop:bounded-order-distance} proves Eq.~\eqref{eq:purifiedFaithfulness}.  Finally, the Fuchs--van de Graaf inequality gives
\begin{equation}
\|\rho-\sigma\|_1\leq2P(\rho,\sigma).
\end{equation}
Applying this at every $d_A$-dimensional reference probe and using the state-optimization formula for the diamond norm of a Hermiticity-preserving map gives Eq.~\eqref{eq:purifiedDiamond}.  The bound is finite even when the stabilized Umegaki divergence is $+\infty$.

\subsection{Reference-dimension reduction}
\begin{proof}
Let $R'$ be an arbitrary finite-dimensional reference and let $\rho_{R'A}$ be any probe state.  Choose a purification $\lvert\Omega\rangle_{ER'A}$ of $\rho_{R'A}$ and define $\mathcal C_0=\Psi\circ\Phi$ and $\mathcal C_1=\Phi\circ\Psi$.  Data processing of $\sandD_\alpha$ under the partial trace over $E$ gives
\begin{align}
&J_\alpha\!\left(
[\id_{R'}\otimes\mathcal C_0](\rho_{R'A}),
[\id_{R'}\otimes\mathcal C_1](\rho_{R'A})
\right)\nonumber\\
&\quad\leq
J_\alpha\!\left(
[\id_{ER'}\otimes\mathcal C_0](\Omega_{ER'A}),
[\id_{ER'}\otimes\mathcal C_1](\Omega_{ER'A})
\right),
\label{eq:purification-bound}
\end{align}
where $\Omega_{ER'A}=\lvert\Omega\rangle\!\langle\Omega\rvert$.  Its Schmidt rank across $ER'|A$ is at most $d_A$, so $\lvert\Omega\rangle=(V\otimes\mathbf 1_A)\lvert\psi\rangle$ for a $d_A$-dimensional reference $R$ and an isometry $V:R\to ER'$.  Isometric invariance identifies the right-hand side of Eq.~\eqref{eq:purification-bound} with the $RA$ value.  Thus $d_A$-dimensional pure probes attain the unrestricted supremum.
\end{proof}

\subsection{State-level inherited properties}
Let $\rho_{12}=(\Psi\circ\Phi)(\rho)$ and $\rho_{21}=(\Phi\circ\Psi)(\rho)$.  For every $\alpha\geq1/2$: (i) $\mathcal H_\alpha\geq0$ with equality exactly when $\rho_{12}=\rho_{21}$; (ii) common CPTP postprocessing contracts $J_\alpha$; (iii) simultaneous unitary conjugation of the input and both channels leaves $\mathcal H_\alpha$ invariant; and (iv) for product inputs and product channel pairs,
\begin{equation}
\mathcal H_\alpha(\Phi_1\otimes\Phi_2,\Psi_1\otimes\Psi_2\mid\rho_1\otimes\rho_2)
=\mathcal H_\alpha(\Phi_1,\Psi_1\mid\rho_1)
+\mathcal H_\alpha(\Phi_2,\Psi_2\mid\rho_2).
\label{eq:appAdditivity}
\end{equation}
\begin{proof}
Faithfulness gives (i), data processing of both directed terms gives (ii), and unitary invariance gives (iii) since the ordered outputs transform by $\operatorname{Ad}_U$.  For (iv),
\begin{align}
[(\Psi_1\otimes\Psi_2)\circ(\Phi_1\otimes\Phi_2)](\rho_1\otimes\rho_2)
&=\rho_{12}^{(1)}\otimes\rho_{12}^{(2)},\\
[(\Phi_1\otimes\Phi_2)\circ(\Psi_1\otimes\Psi_2)](\rho_1\otimes\rho_2)
&=\rho_{21}^{(1)}\otimes\rho_{21}^{(2)}.
\end{align}
Additivity of each directed sandwiched R\'enyi divergence proves Eq.~\eqref{eq:appAdditivity}.
\end{proof}

\subsection{Stabilized channel faithfulness and Umegaki diamond estimate}
\begin{proof}
If $\widehat{\cH}_\alpha=0$, product probes imply equality of the ordered composites on every density operator and hence, by linearity, as channels; the converse is immediate.  Data processing under $\id_R\otimes\Lambda$ proves Eq.~\eqref{eq:stabilized-postprocess}.  For the diamond estimate, Proposition~\ref{prop:reference-reduction} and state optimization for a Hermiticity-preserving map give
\begin{equation}
\norm{\Delta_{\Phi,\Psi}}_\diamond
=\sup_{\omega_{RA}\in\cD(\mathcal H_R\otimes\mathcal H_A)}
\norm{(\id_R\otimes\Delta_{\Phi,\Psi})(\omega_{RA})}_1.
\end{equation}
Corollary~\ref{cor:pinsker} bounds each integrand by the square root of twice its $J_1$ value; taking the supremum proves Eq.~\eqref{eq:diamond}.  When $\widehat{\cH}_1=+\infty$, Eq.~\eqref{eq:purifiedDiamond} supplies finite global norm control.
\end{proof}

\subsection{Local geometric optimization}
\label{app:geometricoptimization}

Let $\theta\mapsto\rho_\theta$ be a $C^3$ faithful model through $\sigma$, with tangents $X_i=\partial_i\rho_\theta|_0$, and write $\Omega=\Omega_{\Phi,\Psi}^{\sigma\to\tau}$.  The pullback tensor is
\begin{equation}
G^{\rm ord}_{ij}=\langle\Omega X_i,\Omega X_j\rangle_{\tau,{\rm BKM}}.
\end{equation}
As a Gram matrix it is positive semidefinite, with null tangents $v^iX_i\in\ker\Omega$; under reparameterization $X'_a=J^i{}_aX_i$ and $G'_{ab}=J^i{}_aJ^j{}_bG_{ij}$.  Common CPTP postprocessing obeys
\begin{equation}
v^iv^j(G^{\rm ord}_{\Lambda})_{ij}
=\|\Lambda\Omega(v^iX_i)\|_{\Lambda(\tau),{\rm BKM}}^2
\le\|\Omega(v^iX_i)\|_{\tau,{\rm BKM}}^2,
\end{equation}
hence $G^{\rm ord}_{\Lambda}\preceq G^{\rm ord}$.

Let $\Omega^\sharp$ be the adjoint with respect to the input and output BKM forms.  On the finite-dimensional real tangent space,
\begin{equation}
\max_{\|X\|_{\sigma,{\rm BKM}}=1}\chi^{\rm ord}(X)
=\|\Omega\|_{{\rm BKM}\to{\rm BKM}}^2
=\lambda_{\max}(\Omega^\sharp\Omega).
\label{eq:appOptimalBKM}
\end{equation}
Rayleigh--Ritz identifies the maximizers with the top eigenspace; the optimum vanishes exactly when $\Omega=0$, i.e. when the composites coincide.  Downstream BKM contraction contracts Eq.~\eqref{eq:appOptimalBKM}.

If linear first-order resource constraints define a BKM-closed subspace $\mathsf K\subset\mathsf T_\sigma$, with BKM projector $\Pi_{\mathsf K}$, the accessible optimum is
\begin{equation}
\max_{X\in\mathsf K,\ \|X\|_{\sigma,{\rm BKM}}=1}\chi^{\rm ord}(X)
=\lambda_{\max}\!\left(
\left.\Pi_{\mathsf K}\Omega^\sharp\Omega\Pi_{\mathsf K}\right|_{\mathsf K}
\right).
\label{eq:appConstrainedBKM}
\end{equation}
An energy-neutral constraint $\operatorname{Tr}(HX)=0$ is one example.

\subsection{Two-base-point BKM expansion}
\label{app:BKM}

Let $A_\eps=\sigma+\eps X$ and $B_\eps=\sigma+\eps Y$, with $\Tr X=\Tr Y=0$, and put $\lambda_*=\lambda_{\min}(\sigma)>0$.  Choose $\eps_0>0$ so that
\begin{equation}
|\eps|\leq\eps_0
\quad\Longrightarrow\quad
A_\eps\succeq \frac{\lambda_*}{2}I,
\qquad
B_\eps\succeq \frac{\lambda_*}{2}I.
\label{eq:faithfulcompact}
\end{equation}
For instance it is enough to take
$\eps_0\max\{\norm{X}_\infty,\norm{Y}_\infty\}\leq\lambda_*/2$.
On this faithful interval the logarithm is analytic, with resolvent derivatives
\begin{align}
\frac{\dd}{\dd\eps}\log A_\eps
&=\int_0^\infty R_A X R_A\,\dd t,\nonumber\\
\frac{\dd^2}{\dd\eps^2}\log A_\eps
&=-2\int_0^\infty R_A X R_A X R_A\,\dd t,\nonumber\\
\frac{\dd^3}{\dd\eps^3}\log A_\eps
&=6\int_0^\infty R_A X R_A X R_A X R_A\,\dd t,
\label{eq:logderivatives}
\end{align}
where $R_A=(A_\eps+tI)^{-1}$; the same formulas hold for $B_\eps$ with $X$ replaced by $Y$.  Equation~\eqref{eq:faithfulcompact} implies
\begin{equation}
\norm{R_A}_\infty,\norm{R_B}_\infty
\leq(t+\lambda_*/2)^{-1}.
\end{equation}
Hence these derivatives are uniformly bounded on $[-\eps_0,\eps_0]$ by constants scaling as $(\lambda_*/2)^{-1}$, $(\lambda_*/2)^{-2}$, and $(\lambda_*/2)^{-3}$.

Define
\begin{equation}
F(\eps)=D(A_\eps\Vert B_\eps)
=\Tr A_\eps(\log A_\eps-\log B_\eps).
\end{equation}
Thus $F\in C^3[-\eps_0,\eps_0]$ and
\begin{equation}
C_*:=\frac16\max_{|\eps|\leq\eps_0}|F^{(3)}(\eps)|<\infty.
\label{eq:thirdbound}
\end{equation}
At $\eps=0$, $F(0)=F'(0)=0$.  The Fr\'echet derivative
\begin{equation}
\left.\frac{\dd}{\dd\eps}\ln(\sigma+\eps Z)\right|_{\eps=0}
=\cT_\sigma(Z)
\end{equation}
and direct differentiation give
\begin{equation}
F''(0)=\Tr[(X-Y)\cT_\sigma(X-Y)]
=\norm{X-Y}_{\sigma,\mathrm{BKM}}^2.
\label{eq:BKMsecondderivative}
\end{equation}
Taylor's theorem and Eq.~\eqref{eq:thirdbound} yield
\begin{equation}
\left|
F(\eps)-\frac{\eps^2}{2}\norm{X-Y}_{\sigma,\mathrm{BKM}}^2
\right|
\leq C_*|\eps|^3,
\end{equation}
which proves Lemma~\ref{lem:twobase}.

In an eigenbasis $\sigma=\sum_a p_a|a\rangle\langle a|$,
\begin{equation}
\norm{Z}_{\sigma,\mathrm{BKM}}^2
=\sum_{a,b}\frac{\ln p_a-\ln p_b}{p_a-p_b}|Z_{ab}|^2,
\label{eq:BKMcomponents}
\end{equation}
with the diagonal coefficient understood as $1/p_a$.

\subsection{Conditional operator-algebraic extension}
\label{app:vN}
The finite-dimensional BKM theorem has a controlled operator-algebraic extension under an explicit regular-family hypothesis.  Let $\mathfrak M$ be a $\sigma$-finite von Neumann algebra, let $\Phi_*,\Psi_*$ be normal predual channels, and let $\omega_\epsilon$ be a $C^3$ normal-state arc through a faithful state.  Suppose the two ordered output arcs meet at a faithful state $\tau$ and lie, for sufficiently small $|\epsilon|$, in one embedded finite-dimensional $C^3$ family of faithful normal states on which Araki relative entropy has the local expansion
\begin{equation}
S_{\mathfrak M}(\rho_\epsilon\Vert\varphi_\epsilon)
=\frac{\epsilon^2}{2}
\|\dot\rho_0-\dot\varphi_0\|_{\tau,\mathrm{BKM}}^2+o(\epsilon^2).
\label{eq:vNtwobase}
\end{equation}
Then predual linearity and the same expansion in both directions give
\begin{equation}
\mathcal H_1^{\mathfrak M}(\Phi_*,\Psi_*\mid\omega_\epsilon)
=\frac{\epsilon^2}{2}
\|[(\Psi_*\Phi_*)-(\Phi_*\Psi_*)]\dot\omega_0\|_{\tau,\mathrm{BKM}}^2
+o(\epsilon^2),
\label{eq:appModularSusceptibility}
\end{equation}
and normal common postprocessing contracts the coefficient whenever the postprocessed arcs satisfy the same hypothesis.  The hypothesis is realized by a finite matrix factor carrying the nontrivial dynamics, $\mathfrak M=M_d(\mathbb C)\bar\otimes\mathfrak N$ with a fixed faithful spectator state on $\mathfrak N$, because additivity reduces the relative entropy to the matrix factor.  This establishes compatibility with type-III ambient algebras when $\mathfrak N$ is type III; an intrinsic differentiable geometry on general type-III state spaces requires further domain and regularity control.

\section{Scattering and Gaussian technical foundations}
\label{app:scatteringFoundations}

This appendix specifies the Regge--Wheeler solver, Gaussian relative-entropy convention, and independent scattering cross-checks supporting Secs.~\ref{sec:gaussian} and \ref{sec:schwarzschild}.
\subsection{Regge--Wheeler numerical method}
\label{app:RW}

The numerical solver evolves $(r,u,\partial_{r_*}u)$ with $r_*$ as the independent variable:
\begin{align}
\frac{\dd r}{\dd r_*}&=f(r),\\
\frac{\dd u}{\dd r_*}&=v,\\
\frac{\dd v}{\dd r_*}&=-[\omega^2-V_\ell(r)]u.
\end{align}
Production initial data are imposed at $r=2M(1+10^{-6})$ with the second-order Frobenius correction of Sec.~\ref{app:frobenius}.  At the outer radius the amplitudes are reconstructed in the second-order Jost basis of Sec.~\ref{app:asymptoticmatching}.  The production matching radius is
\(
r_{\rm match}=\max(700M,160/\omega)
\),
with DOP853 tolerances
\(
\mathrm{rtol}=2\times10^{-11}
\)
and
\(
\mathrm{atol}=5\times10^{-14}
\).
The finite-radius basis currents and the Frobenius horizon current are used in the flux ratios, so the reported residual tests the complete production prescription.

Validation uses tighter domains/tolerances, horizon-start changes, independent areal-radius integration, the analytic infrared limit, quadrature refinement, and transmission-node thinning.  Table~\ref{tab:numericalValidation} gives the largest discrepancies; pointwise records are machine readable.

\begin{table}[t]
\centering
\small
\caption{Independent numerical validation of load-bearing analytic or computational steps.  ``Relative'' and ``absolute'' refer to the discrepancy in the quantity named in the first column.  Full parameter grids and definitions are stored in the cited machine-readable files.}
\label{tab:numericalValidation}
\begin{tabular}{p{0.58\linewidth}p{0.27\linewidth}}
\toprule
Validation check & Largest quoted discrepancy \\
\midrule
Gaussian Fock-space calculation vs. exact one-mode formula & $6.61\times10^{-12}$ nats \\
Regge--Wheeler production vs. tightened solve & $7.94\times10^{-11}$ absolute \\
First- vs. second-order horizon Frobenius data & $1.03\times10^{-10}$ relative \\
Independent $r$- vs. $r_*$-coordinate transmission & $1.13\times10^{-8}$ relative \\
Production Regge--Wheeler flux conservation & $3.53\times10^{-10}$ \\
SdS static-patch flux conservation & $1.93\times10^{-7}$ \\
Finite-memory $2^{21}$ vs. $2^{22}$ FFT readout shift & $8.21\times10^{-8}$ nats \\
Residual-axis analytic identity vs. covariance evaluation & $4.93\times10^{-16}$ nats \\
\bottomrule
\end{tabular}
\end{table}

At $M\omega=10^{-3}$, $\Gamma_0/[16(M\omega)^2]=\sigma_{\rm abs}/(16\pi M^2)=1.00643$.  The transmission/convergence records are in \path{data/regge_wheeler_scalar.csv}, \path{data/rw_convergence_checks.csv}, \path{data/rw_independent_r_coordinate_crosscheck.csv}, and \path{data/numerical_checks.json}; \path{code/generate_figures.py} regenerates the figures.

\subsection{Gaussian relative-entropy derivation}
\label{app:Gaussian}

A squeezed thermal state is
\begin{equation}
\rho(n,s)=U_s\tau_nU_s^\dagger,
\qquad
\tau_n=\frac1{n+1}\left(\frac{n}{n+1}\right)^{a^\dagger a}.
\end{equation}
Writing $b(n)=\ln[(n+1)/n]$ gives
\begin{equation}
\ln\tau_n=-\ln(n+1)-b(n)a^\dagger a.
\end{equation}
By unitary invariance,
\begin{align}
\Tr\rho(n_a,s_a)\ln\rho(n_b,s_b)
&=-\ln(n_b+1)\nonumber\\
&\quad-b(n_b)N_{a|b},
\end{align}
where the mean occupation in the $s_b$ frame is
\begin{equation}
N_{a|b}=\nu_a\cosh[2(s_a-s_b)]-\frac12.
\end{equation}
Using $S[\rho(n_a,s_a)]=g(n_a)$ and
$g(n_b)=\ln(n_b+1)+b(n_b)n_b$ yields Eq.~\eqref{eq:gaussianD}.

For the ordered covariances, Eq.~\eqref{eq:nuA} follows from $\nu_A=\sqrt{\det V_A}$ and Eq.~\eqref{eq:sA} from the ratio of diagonal quadrature variances. The code additionally verifies Eq.~\eqref{eq:gaussianD} by finite-Fock-space diagonalization for representative parameters.

\subsection{Independent areal-radius scattering cross-check}
A second solver rewrites the Regge--Wheeler equation using the areal radius $r$ as the independent variable,
\begin{equation}
f^2u_{,rr}+ff'u_{,r}+[\omega^2-V_\ell(r)]u=0,
\qquad f=1-\frac{2M}{r},
\end{equation}
and integrates $(u,u_{,r})$ with RK45.  At the matching radius it reconstructs $u_{,r_*}=fu_{,r}$ and uses the production Jost basis, thereby changing the integration coordinate and integrator while holding the physical boundary prescription fixed.  Across twelve points with $0.003\le M\omega\le0.8$ and $0\le\ell\le2$, the largest relative transmission difference is $1.13\times10^{-8}$.  At $M\omega=10^{-3}$, $\Gamma_\ell/[g_\ell(M\omega)^{2\ell+2}]=1.00643,1.00660,1.00605,1.00758$ for $\ell=0,1,2,3$; the records are in \path{data/rw_independent_r_coordinate_crosscheck.csv} and \path{data/gell_low_frequency_validation.csv}.

\subsection{Second-order Frobenius horizon data}
\label{app:frobenius}
Set $s=(r-2M)/(2M)$ and write the ingoing solution as
\begin{equation}
u=e^{-i\omega r_*}F(s),\qquad F(s)=1+a_1s+a_2s^2+O(s^3).
\end{equation}
Since
\begin{equation}
\frac{\dd}{\dd r_*}=\frac{s}{2M(1+s)}\frac{\dd}{\dd s},
\end{equation}
the reduced amplitude obeys
\begin{equation}
D_*^2F-2i\omega D_*F-V_\ell F=0,
\qquad D_*:=\frac{s}{2M(1+s)}\frac{\dd}{\dd s}.
\end{equation}
Writing $L=\ell(\ell+1)$ and matching the first two powers of $s$ gives
\begin{align}
a_1&=\frac{L+1}{1-4iM\omega},\\
a_2&=\frac{(L+12iM\omega)a_1+L}{4(1-2iM\omega)}.
\label{eq:appFrobenius2Coefficients}
\end{align}
This is the second-order production horizon expansion.  At $s=s_0$ the production initial data are
\begin{align}
u(s_0)&=e^{-i\omega r_*(s_0)}(1+a_1s_0+a_2s_0^2),\\
\frac{\dd u}{\dd r_*}(s_0)
&=e^{-i\omega r_*(s_0)}\left[-i\omega F(s_0)
+\frac{s_0}{2M(1+s_0)}(a_1+2a_2s_0)\right].
\end{align}
The production table uses $s_0=10^{-6}$.  The script \path{code/regge_wheeler_frobenius.py} compares first- and second-order Frobenius data, $s_0=10^{-5}$ and $10^{-6}$ starts, and the legacy leading-horizon/plane-wave prescription:
\begin{align}
\max\frac{|\Gamma_{F_1}-\Gamma_{F_2}|}{\Gamma_{F_2}}&=1.03\times10^{-10},\\
\max\frac{|\Gamma_{F_2,s_0=10^{-5}}-\Gamma_{F_2,s_0=10^{-6}}|}{\Gamma_{F_2,s_0=10^{-6}}}&=9.37\times10^{-12},\\
\max\frac{|\Gamma_{\rm legacy}-\Gamma_{F_2}|}{\Gamma_{F_2}}&=3.48\times10^{-5}.
\end{align}
Thus the second-order near-horizon prescription has subdominant start-point sensitivity.

\subsection{Second-order asymptotic Jost matching}
\label{app:asymptoticmatching}

The production table uses the second-order Jost basis.  For comparison, the leading plane-wave extraction is evaluated on the same solutions and domains.  With $L=\ell(\ell+1)$, take
\begin{equation}
\psi_s(r)=e^{is\omega r_*}
\left(1+\frac{a_{1,s}}r+\frac{a_{2,s}}{r^2}+O(r^{-3})\right),
\qquad s=\pm1.
\end{equation}
Substitution into the scalar Regge--Wheeler equation, using
$f=1-2M/r$ and
$V_\ell=f[L/r^2+2M/r^3]$, determines
\begin{equation}
a_{1,s}=\frac{isL}{2\omega},
\qquad
a_{2,s}=\frac{L(2-L)}{8\omega^2}+\frac{isM}{2\omega}.
\label{eq:appJostCoeffs}
\end{equation}
Amplitudes follow from the $2\times2$ system built from $\psi_\pm$ and their $r_*$ derivatives.  The script \path{code/regge_wheeler_asymptotic_match.py} compares
\begin{equation}
r_{\rm match}^{(1)}=\max(700M,160/\omega),\qquad
r_{\rm match}^{(2)}=\max(1200M,260/\omega),
\end{equation}
on the same twelve-point audit set.

Domain enlargement changes transmission by $4.07\times10^{-5}$ with plane waves and $1.67\times10^{-9}$ with the corrected basis; the smaller-domain basis shift is $3.86\times10^{-5}$ and the corrected normalized Wronskian-current residual is $8.48\times10^{-8}$.  These values are archived in \path{data/rw_asymptotic_matching_crosscheck.csv}; the corrected basis is the production prescription.

\section{Gaussian and Schwarzschild infrared proofs}
\label{app:centralProofs}

This appendix proves the exact Gaussian order formulas, exponent-transfer theorem, Schwarzschild low-frequency coefficient, finite-amplitude infrared hierarchy, and support-boundary results used in the main text.

\label{app:centralproofs}

\subsection{Exact Gaussian order divergence}
For the two ordered centered Gaussian states, Eq.~\eqref{eq:gaussianD} gives
\begin{align}
D(\rho_A\Vert\rho_B)
&=g(n_B)-g(n_A)+b(n_B)(\nu_A C-\nu_B),\\
D(\rho_B\Vert\rho_A)
&=g(n_A)-g(n_B)+b(n_A)(\nu_B C-\nu_A).
\end{align}
Their arithmetic mean is Eq.~\eqref{eq:exactGaussianHolonomy}.  Faithfulness follows from $\nu_B=m>1/2$ and $\nu_A\ge m$.  The symmetric Umegaki divergence vanishes exactly when the two faithful states coincide; for centered Gaussian states this is equivalent to $V_A=V_B$, and Eq.~\eqref{eq:GaussianEqualityCondition} gives the stated zero set.

\subsection{Weak-squeezing susceptibility}
Main-text Eqs.~\eqref{eq:nuA} and \eqref{eq:sA} give
\begin{align}
\delta\nu:=\nu_A-m&=\frac{2ac}{m}r^2+O(r^4),\\
\delta s:=s_A-r&=-\frac{c}{m}r+O(r^3),\\
C&=1+\frac{2c^2}{m^2}r^2+O(r^4).
\end{align}
Since $n_A-n_B=O(r^2)$ and $b$ is analytic at $n_B=m-1/2>0$, the two directed contributions are
\begin{align}
b(n_B)(\nu_A C-\nu_B)
&=b_B\left(\delta\nu+\frac{2c^2}{m}r^2\right)+O(r^4),\\
b(n_A)(\nu_B C-\nu_A)
&=b_B\left(-\delta\nu+\frac{2c^2}{m}r^2\right)+O(r^4).
\end{align}
The $O(r^2)$ symplectic-eigenvalue shifts cancel under symmetrization, leaving Eq.~\eqref{eq:weakr}.  Evenness under $r\mapsto-r$ makes the next possible term quartic, and two differentiations give Eq.~\eqref{eq:chiord}.

\subsection{Greybody--occupation exponent transfer including finite transmission}
Let $\Gamma(x)=g x^\alpha[1+o(1)]$ and $n(x)=\kappa x^{-\beta}[1+o(1)]$ with $g,\kappa>0$, $\alpha\ge0$, and $\beta\ge0$.  Vacuum attenuation gives
\begin{equation}
y:=\Gamma n,\qquad
m=y+\frac12,\qquad
c=\frac{1-\Gamma}{2},
\end{equation}
and the exact weak-squeezing susceptibility is
\begin{equation}
\chi(x)=4b(y)\frac{c^2}{m},
\qquad b(y)=\ln(1+1/y).
\label{eq:appGenericExactChi}
\end{equation}
For $\alpha>0$, $c=1/2+o(1)$ and $y=g\kappa x^d[1+o(1)]$, $d=\alpha-\beta$.  If $d>0$, then $y\to0$, $b(y)=\ln(1/y)+O(y)$, and $4c^2/m=2+o(1)$, giving
\begin{equation}
\chi=2d\ln(1/x)+2\ln[1/(g\kappa)]+o(1).
\end{equation}
If $d=0$, continuity of Eq.~\eqref{eq:appGenericExactChi} gives
\begin{equation}
\chi\to\frac{\ln(1+1/(g\kappa))}{g\kappa+1/2}.
\end{equation}
If $d<0$, then $y\to\infty$, $b(y)=y^{-1}+O(y^{-2})$, and
\begin{equation}
\chi=y^{-2}[1+o(1)]
=\frac{x^{-2d}}{(g\kappa)^2}[1+o(1)].
\end{equation}

For $\alpha=0$, assume $0<g<1$.  If $\beta>0$, then $\Gamma\to g$, $y=g\kappa x^{-\beta}[1+o(1)]\to\infty$, $c\to(1-g)/2$, and
\begin{equation}
\chi=\frac{(1-g)^2}{(g\kappa)^2}x^{2\beta}[1+o(1)],
\end{equation}
which proves Eq.~\eqref{eq:genericFiniteTransmissionGrowingOccupation}.  If $\beta=0$, direct continuity in Eq.~\eqref{eq:appGenericExactChi} yields
\begin{equation}
\chi\to
\frac{(1-g)^2}{g\kappa+1/2}\ln\!\left(1+\frac1{g\kappa}\right),
\end{equation}
proving Eq.~\eqref{eq:genericFiniteTransmissionFiniteOccupation}.  This completes every weak-squeezing branch in Theorem~\ref{thm:genericExponentLaw}.

For the finite-amplitude logarithmic law assume $\alpha>0$, $d>0$, and $\beta>0$.  Then $n(x)\to\infty$, $y\to0$, and $\Gamma/y=1/n=o(1)$.  The compact-uniform calculation below applies with $g_\ell x^{2\ell+1}/(8\pi)$ replaced by $g\kappa x^d$: uniformly for $|r|\le R$,
\begin{align}
n_A&=y\cosh(2r)+o_R(y),&
C&=\cosh(2r)+O_R(y),\\
b(n_A)&=\ln(1/y)-\ln\cosh(2r)+o_R(1).
\end{align}
Substitution in the exact symmetrized Gaussian relative entropy gives Eq.~\eqref{eq:genericExponentFinite}.  This proves Theorem~\ref{thm:genericExponentLaw}.

\subsection{Matched low-frequency Schwarzschild coefficient}

This subsection proves Lemma~\ref{lem:gellSchwarzschild} and fixes the normalization entering Theorem~\ref{thm:partialwaveIR}.  Write the separated radial field as $R_\ell=u_\ell/r$.  At zero frequency, the Schwarzschild radial equation becomes
\begin{equation}
\frac{d}{dr}\left[r(r-2M)\frac{dR_\ell}{dr}\right]
-\ell(\ell+1)R_\ell=0.
\label{eq:appZeroFreqRadial}
\end{equation}
Introduce
\begin{equation}
z=\frac{r}{M}-1,
\label{eq:appLegendreVariable}
\end{equation}
so that $r=M(z+1)$, $r-2M=M(z-1)$, and $d/dr=M^{-1}d/dz$.  Equation~\eqref{eq:appZeroFreqRadial} therefore becomes
\begin{equation}
\frac{d}{dz}\left[(z^2-1)\frac{dR_\ell}{dz}\right]
-\ell(\ell+1)R_\ell=0.
\label{eq:appLegendreEquation}
\end{equation}
Horizon regularity at $r=2M$, or $z=1$, selects $P_\ell(z)$ with $P_\ell(1)=1$.  The unit-transmitted Jost normalization $u_\ell\sim e^{-i\omega r_*}$ at the future horizon approaches $u_\ell\to1$ in the overlap limit.  Since $u_\ell=rR_\ell$, the corresponding near-region normalization is
\begin{equation}
R_\ell(r)\sim\frac{1}{2M}P_\ell\!\left(\frac rM-1\right).
\label{eq:appNearOverlapNormalization}
\end{equation}
For $M\ll r\ll\omega^{-1}$, the leading large-$z$ term
$P_\ell(z)\sim (2\ell)!\,z^\ell/[2^\ell(\ell!)^2]$ gives
\begin{equation}
u_\ell(r)
\sim
\frac{(2\ell)!}{2^{\ell+1}(\ell!)^2M^{\ell+1}}\,r^{\ell+1}.
\label{eq:appNearOverlapPower}
\end{equation}
In the far region the solution regular in the overlap is $u_\ell=C_\ell rj_\ell(\omega r)$.  Using $j_\ell(x)\sim x^\ell/(2\ell+1)!!$ for $x\ll1$ and matching Eq.~\eqref{eq:appNearOverlapPower} yields
\begin{equation}
C_\ell
=
\frac{(2\ell)!(2\ell+1)!!}{2^{\ell+1}(\ell!)^2}
\frac{1}{M^{\ell+1}\omega^\ell}
\,[1+o(1)].
\label{eq:appFarNormalization}
\end{equation}
At large radius, $rj_\ell(\omega r)\sim \omega^{-1}\sin(\omega r-\ell\pi/2)$.  Replacing $r$ by the Schwarzschild tortoise coordinate $r_*=r+2M\ln(r/2M-1)$ changes the asymptotic Jost phases through the long-range logarithmic phase and leaves the leading incoming magnitude unchanged.  Hence
\begin{equation}
|A_{\rm in}|=\frac{|C_\ell|}{2\omega}\,[1+o(1)].
\label{eq:appIncomingMagnitude}
\end{equation}
With unit transmitted amplitude, flux conservation gives $\Gamma_\ell=|A_{\rm in}|^{-2}$.  Using $(2\ell+1)!!=(2\ell+1)!/(2^\ell\ell!)$ then gives
\begin{equation}
\Gamma_\ell(\omega)
=g_\ell(M\omega)^{2\ell+2}[1+o(1)],
\qquad
g_\ell=
\left[
\frac{2^{2\ell+2}(\ell!)^3}{(2\ell)!(2\ell+1)!}
\right]^2,
\label{eq:appSchwarzschildCoefficient}
\end{equation}
which proves Eq.~\eqref{eq:gell} and supplies the coefficient used in Theorem~\ref{thm:partialwaveIR}.

\subsection{Finite-amplitude infrared hierarchy}
Set $x=M\omega$ and $y=\Gamma_\ell n_H$.  Main-text Eqs.~\eqref{eq:Hawkingoccupation} and \eqref{eq:gell} give
\begin{equation}
y=\frac{g_\ell}{8\pi}x^{2\ell+1}[1+o(1)],
\qquad
\ln\frac1y=L_\ell(x)+o(1).
\end{equation}
For vacuum attenuation, $m=1/2+y$ and $c=(1-\Gamma_\ell)/2$.  Main-text Eq.~\eqref{eq:chiord} therefore gives
\begin{equation}
\chi_\ell=4b(y)\frac{c^2}{m},\qquad
4\frac{c^2}{m}=2+o(1),\qquad
b(y)=\ln(1/y)+O(y),
\end{equation}
which proves Eq.~\eqref{eq:chiellIR}.  For the finite-amplitude law, on every compact $|r|\le R$,
\begin{align}
a&=y+O(xy),& c&=\frac12+O(xy),&m&=\frac12+y,\\
n_A&=y\cosh(2r)+O_R(xy+y^2),\\
s_A&=y\sinh(2r)+O_R(xy+y^2),\\
C&=\cosh(2r)+O_R(y).
\end{align}
Thus
\begin{align}
\nu_A C-\nu_B&=\sinh^2r+O_R(y),\\
\nu_B C-\nu_A&=\sinh^2r+O_R(y),\\
b(n_B)&=\ln(1/y)+O(y),\\
b(n_A)&=\ln(1/y)-\ln\cosh(2r)+O_R(x+y).
\end{align}
Substitution in Eq.~\eqref{eq:exactGaussianHolonomy} leaves
$O_R(y\ln(1/y)+x+y)=o(1)$ and proves Eq.~\eqref{eq:hellIR} uniformly on compact $r$ intervals.  Taking $r\to0$ first gives the sequential limit.  The same compact-uniform estimates imply $\mathfrak h_\ell\le C_R[1+\ln(1/x)]$ near zero, proving fixed-$\ell$ integrability.

\subsection{Infrared R\'enyi threshold and bounded-distance saturation}
For fixed $r\ne0$, the ordered covariances satisfy
\begin{align}
V_A&=\operatorname{diag}(ae^{2r}+c,ae^{-2r}+c)\longrightarrow\tfrac12 I,\\
V_B&=m\operatorname{diag}(e^{2r},e^{-2r})\longrightarrow\tfrac12\operatorname{diag}(e^{2r},e^{-2r}),
\end{align}
because $a\to0$, $c\to1/2$, and $m=1/2+\Gamma n_H\to1/2$.  Write $\rho(n,s)=S(s)\tau_nS^\dagger(s)$.  Then $\rho_B=S(r)\tau_yS^\dagger(r)$ with $y=\Gamma n_H\to0$, while $n_A=\nu_A-1/2\to0$ and $s_A\to0$.  Since
\begin{equation}
\|\tau_n-|0\rangle\langle0|\|_1=\frac{2n}{1+n},
\qquad
\|S(s)|0\rangle\langle0|S^\dagger(s)-|0\rangle\langle0|\|_1
=2\sqrt{1-\operatorname{sech}s},
\label{eq:appGaussianTraceNormBoundary}
\end{equation}
unitary invariance and the triangle inequality give the trace-norm limits in Eq.~\eqref{eq:pureInfraredOrderedStates}.  Their root fidelity is
\begin{equation}
|\langle0|S(r)|0\rangle|=(\cosh r)^{-1/2}.
\end{equation}
For $1/2\le q<1$, the sandwiched R\'enyi divergence of two nonorthogonal pure states is
\begin{equation}
\widetilde D_q(|\psi\rangle\langle\psi|\Vert|\varphi\rangle\langle\varphi|)
=\frac{2q}{q-1}\ln|\langle\varphi|\psi\rangle|,
\end{equation}
so both directions give $q(1-q)^{-1}\ln\cosh r$.  For the finite-$y$ boundary proxy $\rho_0=|0\rangle\langle0|$, $\sigma_y=S(r)\tau_yS^\dagger(r)$, with $z=y/(1+y)$ and $p=(1-q)/q$, the sandwich has rank one:
\begin{equation}
\widetilde Q_q(\rho_0\Vert\sigma_y)
=\bigl[\langle0|\sigma_y^p|0\rangle\bigr]^q.
\end{equation}
The squeezed-vacuum number probabilities and $\sum_{k\ge0}\binom{2k}{k}t^k/4^k=(1-t)^{-1/2}$ give
\begin{equation}
\widetilde Q_q(\rho_0\Vert\sigma_y)
=(\operatorname{sech}r)^q(1-z)^{1-q}
\left[1-z^{2(1-q)/q}\tanh^2r\right]^{-q/2}.
\label{eq:appRenyiBoundaryProxy}
\end{equation}
Thus $\widetilde Q_q\to(\operatorname{sech}r)^q$ with correction exponent $\mu_q=\min\{1,2(1-q)/q\}>0$.  For $P_0=|0\rangle\langle0|$ and $P_r=S(r)P_0S^\dagger(r)$, positive powers preserve the projections, so
\begin{align}
\widetilde Q_q(P_0\Vert P_r)
&=\|P_0P_r\|_{2q}^{2q}
=|\langle0|S(r)|0\rangle|^{2q}
=(\operatorname{sech}r)^q,\nonumber\\
\widetilde Q_q(P_r\Vert P_0)&=(\operatorname{sech}r)^q.
\label{eq:appPureRenyiExact}
\end{align}
For the faithful ordered states, Schatten-norm continuity justifies passage to this boundary pair.  Fix $1/2\le q<1$ and put
\begin{equation}
p_q:=2q\ge1,\qquad
a_q:=\frac{1-q}{2q},\qquad
r_q:=\frac1{a_q}=\frac{2q}{1-q}.
\end{equation}
The sandwiched quasi-entropy can be written as
\begin{equation}
\widetilde Q_q(\rho\Vert\sigma)
=\left\|\rho^{1/2}\sigma^{a_q}\right\|_{p_q}^{p_q}.
\label{eq:appRenyiSchattenRepresentation}
\end{equation}
For density operators, the square-root estimate
\begin{equation}
\|\rho_j^{1/2}-\rho^{1/2}\|_2^2
\le \|\rho_j-\rho\|_1
\end{equation}
and the fractional-power Schatten estimate
\begin{equation}
\|\sigma_j^{a_q}-\sigma^{a_q}\|_{r_q}
\le C_q\|\sigma_j-\sigma\|_1^{a_q}
\label{eq:appFractionalPowerHolder}
\end{equation}
show product convergence in $S_{p_q}$; the fractional-power estimate is the Birman--Koplienko--Solomyak/Ando operator-H\"older inequality in Schatten classes, in the form of Ref.~\cite{Ricard2018FractionalPowers}.  Schatten H\"older gives
\begin{align}
&\|\rho_j^{1/2}\sigma_j^{a_q}-\rho^{1/2}\sigma^{a_q}\|_{p_q}\nonumber\\
&\quad\le
\|\rho_j^{1/2}-\rho^{1/2}\|_2\|\sigma_j^{a_q}\|_{r_q}
+\|\rho^{1/2}\|_2\|\sigma_j^{a_q}-\sigma^{a_q}\|_{r_q}
\longrightarrow0.
\label{eq:appRenyiProductContinuity}
\end{align}
Since $p_q\ge1$, Eq.~\eqref{eq:appRenyiSchattenRepresentation} is jointly continuous at this nonorthogonal pure-state pair, giving
\begin{align}
\widetilde Q_q(\rho_A\Vert\rho_B)
&\longrightarrow(\operatorname{sech}r)^q,\nonumber\\
\widetilde Q_q(\rho_B\Vert\rho_A)
&\longrightarrow(\operatorname{sech}r)^q,
\label{eq:appGaussianRenyiBoundaryQ}
\end{align}
and therefore
\begin{equation}
\widetilde D_q(\rho_A\Vert\rho_B),\ 
\widetilde D_q(\rho_B\Vert\rho_A)
\longrightarrow \frac{q}{1-q}\ln\cosh r.
\end{equation}
This proves Eq.~\eqref{eq:renyiInfraredFinite}.  The finite-$y$ Gaussian formula of Ref.~\cite{SeshadreesanLamiWilde2018} and independent Fock calculations for $q=1/2,3/4,4/5$ through $N_{\rm Fock}=160$ give the same limit; data are archived in \path{data/renyi_support_boundary.csv} and \path{data/renyi_cutoff_convergence.csv}.  At $q=1$, Eq.~\eqref{eq:umegakiInfraredBoundary} follows from the finite-amplitude law.

For $q>1$, let $p_A(y)$ and $p_B(y)$ be the outcome distributions of the measurement $\{P_0,I-P_0\}$ on $\rho_A$ and $\rho_B$.  As $y\to0$, $p_A(y)\to(1,0)$ and $p_B(y)\to(\operatorname{sech}r,1-\operatorname{sech}r)$.  The second component of $p_B(y)$ tends to a positive number while the corresponding component of $p_A(y)$ tends to zero; hence the term $p_{B,2}(y)^q p_{A,2}(y)^{1-q}$ diverges and the reverse directed classical R\'enyi divergence obeys $\widetilde D_q(p_B(y)\Vert p_A(y))\to+\infty$.  Data processing then implies $\widetilde D_q(\rho_B\Vert\rho_A)\to+\infty$, which forces the symmetrized divergence in Eq.~\eqref{eq:renyiInfraredDivergent} to diverge.

For pure states the purified and trace distances equal $\sqrt{1-|\langle\psi|\varphi\rangle|^2}$, giving Eq.~\eqref{eq:boundedInfraredDistances}.  Zero-mean homodyne outcomes are Gaussian, with
\begin{equation}
J_{\rm G}(v_1,v_2)=\frac14\left(\frac{v_1}{v_2}+\frac{v_2}{v_1}-2\right).
\end{equation}
The infrared variance ratios are $e^{\pm2r}$, yielding $J_{\rm G}=\sinh^2r$ and Eq.~\eqref{eq:homodyneInfraredSaturation}.

\section{Schwarzschild--de Sitter infrared asymptotics}
\label{app:sdsProofs}

This appendix derives the source-isolated and product-horizon thermal SdS limits, their double-scaling crossover, and the direct radial-equation checks supporting Proposition~\ref{prop:sdsTwoHorizonScaling}.
\subsection{Source-isolated and product-horizon thermal infrared regimes}
Let $a=r_h/r_c$ and $\varepsilon=\omega r_h$.  The two positive horizon temperatures in the static patch give
\begin{align}
\vartheta_h:=T_hr_h
&=\frac{(1-a)(1+2a)}{4\pi(1+a+a^2)},\nonumber\\
\vartheta_c:=T_cr_h
&=\frac{a(1-a)(2+a)}{4\pi(1+a+a^2)},
\label{eq:appSdsTemperatures}
\end{align}
which follows from $T_i=|f'(r_i)|/(4\pi)$.  The displayed static coordinate fixes the normalization of the Killing generator used for $\omega$, $T_h$, and $T_c$.  A common constant rescaling changes all three by the same factor, leaving $\varepsilon/\vartheta_{h,c}$ and all infrared branch exponents invariant; quoted coefficients use the displayed metric normalization.  We use the product-thermal two-port benchmark
\begin{equation}
n_h=\frac1{e^{\varepsilon/\vartheta_h}-1},
\qquad
n_c=\frac1{e^{\varepsilon/\vartheta_c}-1}.
\end{equation}
Each marginal is fixed by its own horizon temperature.  The two-temperature structure mirrors the factorized horizon KMS data used by Brum and Jor\'as in their invariant Hadamard, non-KMS construction on the relevant SdS region \cite{BrumJoras2015}.  We use the displayed marginals as a stationary nonequilibrium quasifree scattering benchmark.  Global Hartle--Hawking--Israel, Unruh, and single-temperature KMS constructions require state definitions beyond this benchmark.  For transmissivity $\eta$ from the event-horizon port to the selected output, the generic Gaussian parameters are
\begin{equation}
N:=m-\frac12=\eta n_h+(1-\eta)n_c,
\qquad
c=(1-\eta)(n_c+1/2),
\end{equation}
and the exact weak-squeezing susceptibility is
\begin{equation}
\chi=4\ln\!\left(1+\frac1N\right)\frac{c^2}{N+1/2}.
\label{eq:appSdsTwoPortExact}
\end{equation}

For a minimally coupled massless scalar $s$ wave, the strict zero-frequency transmission is
\begin{equation}
\eta\longrightarrow g_{\rm dS}
=\frac{4a^2}{(1+a^2)^2},
\label{eq:appSdsG0}
\end{equation}
throughout the nonextremal range \cite{deCesareMirandaPorfyriadis2026,CrispinoHiguchiOliveiraRocha2013}.  For the source-isolated incoming-port prescription $n_c=0$, $N=g_{\rm dS}\vartheta_h/\varepsilon+O(1)$ and $c\to(1-g_{\rm dS})/2$.  This diagnostic boundary condition isolates the event-horizon source term; a globally regular SdS equilibrium-state construction lies outside this diagnostic use.  The large-$N$ expansion of Eq.~\eqref{eq:appSdsTwoPortExact} gives
\begin{equation}
\chi_{0,\mathrm{src}}^{\rm SdS}
=\frac{(1-g_{\rm dS})^2}{g_{\rm dS}^2\vartheta_h^2}
\varepsilon^2[1+o(1)],
\end{equation}
proving Eq.~\eqref{eq:sdsSuppressedChi} and realizing the $\alpha=0$, $\beta=1$ branch of Theorem~\ref{thm:genericExponentLaw}.

With both thermal ports populated,
\begin{equation}
N=\frac{K_a}{\varepsilon}+O(1),
\quad
c=\frac{(1-g_{\rm dS})\vartheta_c}{\varepsilon}+O(1),
\quad
K_a=g_{\rm dS}\vartheta_h+(1-g_{\rm dS})\vartheta_c.
\end{equation}
Using $\ln(1+1/N)=N^{-1}+O(N^{-2})$ in Eq.~\eqref{eq:appSdsTwoPortExact} yields
\begin{equation}
\chi_{0,\mathrm{2T}}^{\rm SdS}
\longrightarrow
4\left[\frac{(1-g_{\rm dS})\vartheta_c}{K_a}\right]^2,
\end{equation}
which proves Eq.~\eqref{eq:sdsTwoThermalLimit}.  At $a=0.1$ this limit is $2.63017265097$; the smallest point in \path{data/sds_ir_order.csv} agrees to relative accuracy $4.6\times10^{-14}$.

For nonzero scalar-curvature coupling, Ref.~\cite{CrispinoHiguchiOliveiraRocha2013} gives $\eta=g_\xi\varepsilon^2[1+o(1)]$ when the leading coefficient is positive.  For the same source-isolated prescription, $n_c=0$ and $N=g_\xi\vartheta_h\varepsilon[1+o(1)]$, $c=1/2+o(1)$, and Eq.~\eqref{eq:appSdsTwoPortExact} gives
\begin{equation}
\chi_{0,\xi,\mathrm{src}}^{\rm SdS}
=2\ln(1/\varepsilon)+O(1).
\end{equation}
If $n_c$ has its cosmological thermal occupation, then $N=\vartheta_c/\varepsilon+O(1)$ and $c=\vartheta_c/\varepsilon+O(1)$, so
\begin{equation}
\chi_{0,\xi,\mathrm{2T}}^{\rm SdS}\longrightarrow4.
\end{equation}
This proves Proposition~\ref{prop:sdsTwoHorizonScaling} in all four stated sectors.

\paragraph*{Joint small-$a$/small-$\varepsilon$ limit.}
For small black holes, Ref.~\cite{deCesareMirandaPorfyriadis2026} derives the matched expression
\begin{equation}
\Gamma_{0,\rm MA}^{\rm SdS}(a,\varepsilon)
=\frac{4(a^2+\varepsilon^2)}
{a^2\varepsilon^2+(1+a^2+\varepsilon^2)^2},
\label{eq:appSdsMatchedGamma}
\end{equation}
which is symmetric under $a\leftrightarrow\varepsilon$.  For $a,\varepsilon\ll1$,
$\Gamma=4(a^2+\varepsilon^2)[1+o(1)]$ and $n_h=\vartheta_h/\varepsilon+O(1)$.  Under the source-isolated prescription $n_c=0$, $y=\Gamma n_h$ determines four sectors.

If $\varepsilon\ll a^2$, then $y\sim4\vartheta_h a^2/\varepsilon\to\infty$ and
\begin{equation}
\chi\sim\frac{\varepsilon^2}{16\vartheta_h^2a^4}.
\end{equation}
If $a\to0$ with $\varepsilon=\lambda a^2$, $0<\lambda<\infty$, then
$y\to4\vartheta_0/\lambda$, $\vartheta_0=1/(4\pi)$, and
\begin{equation}
\chi\longrightarrow
F(\lambda):=
\frac{\ln[1+\lambda/(4\vartheta_0)]}
{4\vartheta_0/\lambda+1/2}.
\end{equation}
If $a^2\ll\varepsilon\ll a$, then $y\sim4\vartheta_h a^2/\varepsilon\to0$ and
\begin{equation}
\chi=2\ln\!\frac{\varepsilon}{4\vartheta_h a^2}+o(1).
\end{equation}
Finally, if $a\ll\varepsilon\ll1$, then $y\sim4\vartheta_h\varepsilon\to0$ and
\begin{equation}
\chi=2\ln\!\frac1{4\vartheta_h\varepsilon}+o(1),
\end{equation}
which is the Schwarzschild-like branch.  The noncommuting fixed-$a$ and $a\to0$ infrared limits are therefore resolved by the crossover curves $\varepsilon\sim a^2$ and $\varepsilon\sim a$.  The full grid used in Fig.~\ref{fig:sdsDoubleScaling} is stored in \path{data/sds_double_scaling.npz}, with machine-readable slices in \path{data/sds_double_scaling_slices.csv}.

For the product-thermal incoming state the small-$a$ limit has a distinct crossover because $\vartheta_c\sim a/(2\pi)$.  Take
\begin{equation}
 a\to0,\qquad \varepsilon=\lambda a,
 \qquad 0<\lambda<\infty.
\end{equation}
Then $\eta\to0$, $\eta n_h\to0$, and
\begin{equation}
 n_c\longrightarrow \frac1{e^{2\pi\lambda}-1},
 \qquad c\longrightarrow n_c+\frac12.
\end{equation}
Substitution into Eq.~\eqref{eq:appSdsTwoPortExact} gives
\begin{equation}
\chi_{0,\mathrm{2T}}^{\rm SdS}\longrightarrow4\pi\lambda\coth(\pi\lambda).
\label{eq:appSdsTwoThermalDoubleScaling}
\end{equation}
For $\lambda\to0$ this tends to $4$, matching the strongly occupied cosmological-port sector.  For $\varepsilon/a\to\infty$ the cosmological occupation is exponentially suppressed and the source-isolated Schwarzschild-recovery sectors re-emerge; therefore convergence of Eq.~\eqref{eq:appSdsTwoThermalDoubleScaling} is nonuniform at large $\lambda$.  Figure~\ref{fig:appSdsTwoThermalPhase} shows the complete product-thermal grid generated from the exact two-port susceptibility and the matched transmission.  The grid is archived in \path{data/sds_two_thermal_phase.npz}, and fixed-$\lambda$ convergence data are in \path{data/sds_two_thermal_double_scaling.csv}.

\begin{figure}[t]
\centering
\includegraphics[width=0.75\textwidth]{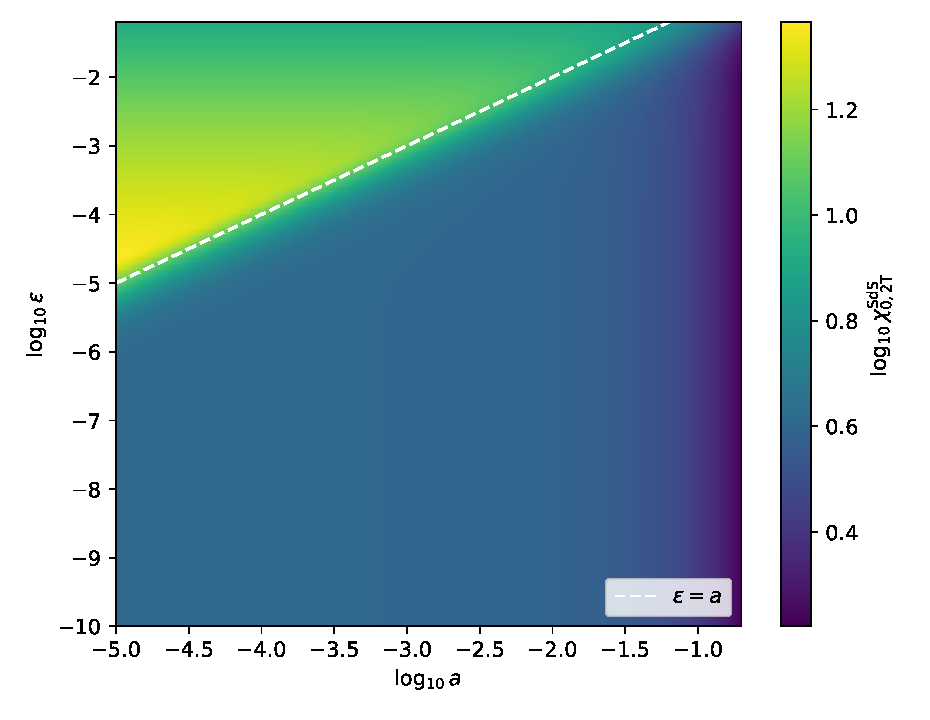}
\caption{Product-horizon thermal SdS susceptibility on the same $(a,\varepsilon)$ domain as the source-isolated phase diagram.  The dashed line marks $\varepsilon=a$, which controls the cosmological-horizon occupation crossover.  At fixed $\lambda=\varepsilon/a$ and $a\to0$, the susceptibility approaches Eq.~\eqref{eq:appSdsTwoThermalDoubleScaling}; for $\varepsilon/a\gg1$ the cosmological input becomes exponentially dilute and the source-isolated sectors are recovered nonuniformly.}
\label{fig:appSdsTwoThermalPhase}
\end{figure}

\paragraph*{Independent direct radial integration.}
The matched expression is tested with a separate static-patch ODE solver.  With $r_h=1$ and $r_c=1/a$,
\begin{equation}
\Lambda=\frac{3}{r_h^2+r_hr_c+r_c^2},
\qquad
M=\frac{r_hr_c(r_h+r_c)}{2(r_h^2+r_hr_c+r_c^2)}.
\end{equation}
For $\phi=e^{-i\omega t}u(r)Y_{00}/r$ and curvature coupling $\xi$,
\begin{equation}
u_{,rr}+\frac{f'}f u_{,r}
+\left[\frac{\omega^2}{f^2}-\frac{f'/r+4\xi\Lambda}{f}\right]u=0.
\label{eq:appSdsDirectODE}
\end{equation}
Second-order Frobenius data are imposed at both horizons.  Write locally
\begin{equation}
f=f_1s+f_2s^2+f_3s^3+O(s^4),\qquad U:=\frac{f'}r+4\xi\Lambda=U_0+U_1s+O(s^2),
\end{equation}
and seek
\begin{equation}
u=s^p(1+A s+B s^2+O(s^3)),\qquad p^2+\frac{\omega^2}{f_1^2}=0.
\end{equation}
The first coefficient is
\begin{equation}
A=\frac{U_0-f_2(p+2p^2)}{f_1(2p+1)},
\label{eq:appSdsFrobeniusOne}
\end{equation}
and direct coefficient matching gives
\begin{multline}
B=\frac{1}{4f_1^2(p+1)}\big[A U_0 f_1-2A f_1f_2p^2-A f_1f_2p-A f_1f_2\\
-U_0f_2+U_1f_1-2f_1f_3p^2-2f_1f_3p+3f_2^2p^2+f_2^2p\big].
\label{eq:appSdsFrobeniusTwo}
\end{multline}
At the cosmological horizon the local coordinate is $s=r_c-r$, so the signs entering $f_1,f_3,U_1$ are changed consistently by that definition.  The sign of $p$ selects the required ingoing/outgoing branch.  The transmission is obtained from the stable conserved radial current $J=f\,\mathrm{Im}(u^*u_{,r})$; with incoming cosmological-horizon amplitude $A_{\rm in}$, $\Gamma_0=(|J_h|/\omega)/|A_{\rm in}|^2$.  This current form avoids subtracting nearly equal asymptotic fluxes.

The production audit spans $a=0.01,0.05,0.1,0.5,0.9$, five frequencies from $\varepsilon=10^{-3}$ to $10^{-5}$, and $\xi=0,1/6$.  For $\xi=0$ the fitted infrared slopes approach zero and the smallest-frequency transmission approaches $4a^2/(1+a^2)^2$; for $\xi=1/6$ the fitted slopes approach $2$.  The $a=0.9$ points are used only as near-Nariai diagnostics.  Table~\ref{tab:numericalValidation} gives the maximum flux residual, while horizon-offset, Frobenius-order, and susceptibility-level checks are stored in \path{data/sds_radial_solver.csv}, \path{data/sds_radial_offset_audit.csv}, and \path{data/sds_radial_frobenius_audit.csv}; the independent solver is \path{code/sds_radial_solver.py}.

\section{Asymptotic ordering and transmission-phase calibration}
\label{app:phaseCalibration}

This appendix fixes the asymptotic input--output meaning of the two channel orders and derives the phase-calibration and residual-axis bounds used in Sec.~\ref{sec:schwarzschild}.

\label{app:asymptoticfactorization}

The physical channel ordering used in the main text is defined on asymptotic Regge--Wheeler in/out modes.  A control with spacetime support overlapping the curvature barrier belongs to a separate interaction model.  For one $(\ell,m,\omega)$ sector, let the reciprocal two-port unitary scattering matrix act as
\begin{equation}
\begin{pmatrix}b^{\rm out}_\omega\\ d^{\rm hor}_\omega\end{pmatrix}
=
\mathbb S_\ell(\omega)
\begin{pmatrix}a^{\rm up}_\omega\\ c^{\rm in}_\omega\end{pmatrix},
\qquad
b^{\rm out}_\omega=t_\ell a^{\rm up}_\omega+r_\ell c^{\rm in}_\omega.
\label{eq:appTwoPortS}
\end{equation}
The exterior input is vacuum.  Writing
$t_\ell=\sqrt{\Gamma_\ell}e^{i\phi_\ell}$ and tracing the complementary output gives the Gaussian channel
\begin{equation}
\mathcal C_{\Gamma,\phi}=\mathcal R_\phi\circ\mathcal L_{\Gamma,0}.
\end{equation}
Because the vacuum attenuator is phase covariant,
$\mathcal R_\phi\mathcal L_{\Gamma,0}=\mathcal L_{\Gamma,0}\mathcal R_\phi$.
For any prescribed input control $\mathcal S_{\rm in}$ define
\begin{equation}
\mathcal S_{\rm out}^{(\phi)}
=\mathcal R_\phi\circ\mathcal S_{\rm in}\circ\mathcal R_\phi^{-1}.
\end{equation}
Then
\begin{align}
\mathcal C_{\Gamma,\phi}\circ\mathcal S_{\rm in}
&=\mathcal R_\phi\circ(\mathcal L_{\Gamma,0}\circ\mathcal S_{\rm in}),\\
\mathcal S_{\rm out}^{(\phi)}\circ\mathcal C_{\Gamma,\phi}
&=\mathcal R_\phi\circ(\mathcal S_{\rm in}\circ\mathcal L_{\Gamma,0}).
\end{align}
Thus any common-unitary-invariant divergence between the two outputs is exactly the divergence between the loss--control and control--loss outputs.  The unitary $\mathcal R_\phi^{-1}$ provides the exact calibrated mode-frame identification from the out-mode Hilbert space back to the in-mode Hilbert space.  Under this identification $\mathcal C_{\Gamma,\phi}$ becomes $\mathcal L_{\Gamma,0}$ and $\mathcal S_{\rm out}^{(\phi)}$ becomes the same channel $\mathcal S_{\rm in}$, so the abstract composition-generated pair is $(\mathcal L_\Gamma\mathcal S,\mathcal S\mathcal L_\Gamma)$ on one calibrated mode space.  On a finite band, $\mathcal R_\phi$ is the diagonal second-quantized phase rotation generated by $\phi_\ell(\omega)$.  A physical finite-band phase compensator approximates this exact frame identification; the affine implementation error and residual-axis penalty are quantified below.

\begin{proposition}[Exact robustness to a residual squeezing-axis mismatch]
\label{prop:appPhaseMismatch}
Consider the faithful centered one-mode Gaussian model of the main text, and work in the calibrated output frame in which the scattering phase has been removed.  Let the post-control squeezing axis retain a residual angle $\delta\phi$ relative to the transported input axis.  With $a,c,m,\nu_A,s_A,n_A,n_B$ defined in Eqs.~\eqref{eq:acm}--\eqref{eq:nuB}, the ordered covariances are
\begin{align}
V_A&=\nu_A\operatorname{diag}(e^{2s_A},e^{-2s_A}),\\
V_B(\delta\phi)&=m R_{\delta\phi}
\operatorname{diag}(e^{2r},e^{-2r})R_{\delta\phi}^{T},
\end{align}
where $R_{\delta\phi}$ is the quadrature rotation.  Their symmetrized Umegaki divergence is
\begin{equation}
\mathcal H_1^{\rm G}(\delta\phi)
=\frac12\left\{b(n_B)[\nu_A C_{\delta}-m]
+b(n_A)[m C_{\delta}-\nu_A]\right\},
\label{eq:appPhaseMismatchExact}
\end{equation}
with
\begin{equation}
C_{\delta}=\cosh(2s_A)\cosh(2r)
-\cos(2\delta\phi)\sinh(2s_A)\sinh(2r).
\end{equation}
Consequently
\begin{align}
\mathcal H_1^{\rm G}(\delta\phi)-\mathcal H_1^{\rm G}(0)
&=K_\phi\sin^2\delta\phi,\\
K_\phi&=[b(n_B)\nu_A+b(n_A)m]
\sinh(2s_A)\sinh(2r)\ge0,
\label{eq:appPhaseMismatchPenalty}
\end{align}
and therefore
\begin{equation}
0\le \mathcal H_1^{\rm G}(\delta\phi)-\mathcal H_1^{\rm G}(0)
\le K_\phi(\delta\phi)^2.
\label{eq:appPhaseMismatchQuadraticBound}
\end{equation}
\end{proposition}
\begin{proof}
Undoing the common scattering rotation puts $V_A$ on its principal axes and rotates only the residual post-control covariance by $\delta\phi$.  In the Gaussian relative-entropy formula, the aligned factor $\cosh[2(s_A-r)]$ is therefore replaced by $C_{\delta}$.  Substitution in both directed relative entropies gives Eq.~\eqref{eq:appPhaseMismatchExact}.  Since
\begin{equation}
C_{\delta}-C_0
=2\sin^2\delta\phi\,\sinh(2s_A)\sinh(2r),
\end{equation}
Eq.~\eqref{eq:appPhaseMismatchPenalty} follows.  The definition of $s_A$ shows that $s_A$ has the same sign as $r$ whenever $a>0$, while for $a=0$ the product vanishes.  Hence $K_\phi\ge0$.  Finally $\sin^2 x\le x^2$ gives Eq.~\eqref{eq:appPhaseMismatchQuadraticBound}.
\end{proof}

\begin{corollary}[Infrared-uniform Schwarzschild phase robustness]
\label{cor:appPhaseMismatchIR}
For a Schwarzschild partial wave with fixed squeezing $r$, put $x=M\omega$ and $y_\ell=\Gamma_\ell n_H$.  As $x\to0^+$,
\begin{equation}
K_\phi=2y_\ell\sinh^2(2r)\ln\frac1{y_\ell}+O_r(y_\ell)
=O_r\!\left[x^{2\ell+1}\ln\frac1x\right].
\label{eq:appPhaseMismatchIR}
\end{equation}
Consequently, for every fixed residual angle $\delta\phi$,
\begin{equation}
0\le \mathcal H_1^{\rm G}(\delta\phi)-\mathcal H_1^{\rm G}(0)
=O_r\!\left[(\delta\phi)^2x^{2\ell+1}\ln\frac1x\right]\longrightarrow0.
\end{equation}
\end{corollary}
\begin{proof}
In the Schwarzschild vacuum-attenuation limit, $m=1/2+y_\ell$ and $\Gamma_\ell=o(y_\ell)$.  The finite-amplitude expansions already used in the infrared hierarchy give
\begin{align}
\nu_A&=\frac12+y_\ell\cosh(2r)+O_r(y_\ell^2+\Gamma_\ell),\nonumber\\
s_A&=y_\ell\sinh(2r)+O_r(y_\ell^2+\Gamma_\ell),\nonumber\\
b(n_B)&=\ln(1/y_\ell)+O(y_\ell),\nonumber\\
b(n_A)&=\ln(1/y_\ell)-\ln\cosh(2r)+O_r(y_\ell+\Gamma_\ell/y_\ell).
\end{align}
Substitution in Eq.~\eqref{eq:appPhaseMismatchPenalty} yields the first equality in Eq.~\eqref{eq:appPhaseMismatchIR}.  Since $y_\ell=(g_\ell/8\pi)x^{2\ell+1}[1+o(1)]$, the stated big-$O$ law follows.  Combining it with Eq.~\eqref{eq:appPhaseMismatchQuadraticBound} proves the fixed-angle result.
\end{proof}

This result separates the calibrated attenuation-order signal from a residual quadrature-axis calibration contribution.  A residual mismatch can generate distinguishability even at unit transmission, where the calibrated attenuation-order response vanishes.  The covariance-level identity is independently checked in \path{code/phase_mismatch_robustness.py} and summarized in Table~\ref{tab:numericalValidation}; the finite-band calculation retains the measured frequency-dependent residual.

\section{Selected-port finite-time receiver and detector bounds}
\label{app:receiverProofs}

This appendix constructs the reduced selected-port capture--squeeze--release controller, its bounded finite-time approximation, the packet covariance bound, and the downstream detector model.
\paragraph*{Asymptotic and bounded finite-time capture--squeeze--release controller.}
Let $b(t)$ be the outgoing chiral field of the selected Markov receiver port, $[b(t),b^\dagger(t')]=\delta(t-t')$, and define
\begin{equation}
A_u=\int_{-\infty}^{\infty}\dd t\,\overline{u(t)}b(t),\qquad \int|u(t)|^2\dd t=1.
\label{eq:appTemporalPacketMode}
\end{equation}
The assumptions are: lossless one-sided memory, Born--Markov coupling over the receiver bandwidth, dispersion-free chiral propagation after the independently audited Jost placement, and programmable complex capture/release coupling.  Under these assumptions standard input--output theory gives an asymptotic capture unitary $C_u$ with
\begin{equation}
C_u^\dagger c C_u=A_u
\label{eq:appCaptureIdentity}
\end{equation}
\cite{KiilerichMolmer2019,NurdinJamesYamamoto2016}.  The familiar ideal couplings have an endpoint singularity; this is the asymptotic target rather than a finite-hardware claim.

\begin{theorem}[Bounded finite-time one-sided capture]
\label{thm:appBoundedCapture}
Let $u_T$ be any normalized temporal mode supported on a finite interval $I=[t_0,t_1]$, and choose $0<\eta<1$.  A phase convention can be chosen so that the bounded coupling
\begin{equation}
|g_\eta(t)|=
\frac{\sqrt\eta\,|u_T(t)|}
{\sqrt{1-\eta+\eta F_T(t)}},
\qquad
F_T(t)=\int_{t_0}^{t}|u_T(s)|^2\dd s,
\label{eq:appBoundedCaptureCoupling}
\end{equation}
implements
\begin{equation}
c(t_1)=\sqrt{1-\eta}\,c(t_0)+\sqrt\eta\,A_{u_T},
\label{eq:appBoundedCaptureMap}
\end{equation}
up to the harmless overall phase fixed by the input--output convention.  In particular,
\begin{equation}
\|g_\eta\|_\infty
\le\sqrt{\frac{\eta}{1-\eta}}\,\|u_T\|_\infty.
\label{eq:appBoundedCaptureGmax}
\end{equation}
\end{theorem}
\begin{proof}
For a one-sided oscillator the Heisenberg--Langevin equation is linear.  Writing $G(t)=\int_t^{t_1}|g(s)|^2\dd s$ gives
\begin{equation}
c(t_1)=e^{-G(t_0)/2}c(t_0)-\int_{t_0}^{t_1}\dd t\,\overline{g(t)}e^{-G(t)/2}b(t).
\end{equation}
The captured mode has squared norm $1-e^{-G(t_0)}$.  Substituting Eq.~\eqref{eq:appBoundedCaptureCoupling} gives $e^{-G(t)}=1-\eta+\eta F_T(t)$, hence Eq.~\eqref{eq:appBoundedCaptureMap}; the denominator is bounded below by $\sqrt{1-\eta}$, which proves Eq.~\eqref{eq:appBoundedCaptureGmax}.
\end{proof}

To obtain a compactly supported target without a hard spectral cut, use the normalized $C^\infty$ mode
\begin{equation}
v_T(t)=\frac{w_T(t)u(t)}{\|w_Tu\|_2},\qquad
w_T(t)=
\begin{cases}
\exp\!\left[1-\dfrac{1}{1-(2t/T)^2}\right],& |t|<T/2,\\
0,& |t|\ge T/2.
\end{cases}
\label{eq:appSmoothCompactMode}
\end{equation}
Theorem~\ref{thm:appBoundedCapture} applies to capture and release of $v_T$ with efficiency $\eta$.  Since $[b(t),b^\dagger(t')]=\delta(t-t')$, the coupling has dimension $[g]=({\rm time})^{-1/2}$ and $\kappa=|g|^2$ has dimension $({\rm time})^{-1}$.  The archived production packet uses $T=3.0\times10^6M$ and $\eta=0.9999$; the complete overlap, out-of-band, coupling, and convergence diagnostics are stored in \path{data/finite_memory_capture.json} and \path{data/finite_memory_hawking_convergence.json}.

For fixed squeezing axis $\theta$, define
\begin{equation}
S_c(r,\theta)=\exp\left\{\frac r2\left[e^{-2i\theta}c^2-e^{2i\theta}c^{\dagger2}\right]\right\}.
\end{equation}
The asymptotic capture identity gives the exact selected-port conjugation
\begin{align}
U_u&:=C_u^\dagger S_c(r,\theta)C_u\nonumber\\
&=\exp\left\{\frac r2\left[e^{-2i\theta}A_u^2-e^{2i\theta}A_u^{\dagger2}\right]\right\}
=S_{A_u}(r,\theta)\otimes I_{u^\perp}.
\label{eq:appMemoryConjugationProof}
\end{align}
The bounded controller approaches this identity as the capture and release efficiencies approach unity.  Using the archived Regge--Wheeler transmission and Hawking occupation, the production packet at $r=0.30$ changes the ideal selected-mode symmetrized Umegaki divergence by
\begin{equation}
|\Delta\mathcal H_{1,u}|=1.991\times10^{-4}\ {\rm nats},
\label{eq:appFiniteMemoryHawkingError}
\end{equation}
with convergence recorded in Table~\ref{tab:numericalValidation}.  This statistic belongs to the reduced selected-$s$-wave receiver and includes the archived mode-deformation, capture/release, thermal-variation, and greybody-dispersion corrections.

\subsection{Packet covariance and detector model}
\label{app:packetdetector}

To prove Eq.~\eqref{eq:packetCovarianceBound}, define the orthogonal residual mode
\begin{equation}
R_f:=\int_B\dd\omega\,\overline{f(\omega)}
[\widetilde t_\ell(\omega)-\tau_f]a_\omega^{\rm up}.
\label{eq:appPacketResidual}
\end{equation}
By Eq.~\eqref{eq:packetLeakage},
$[A_f,R_f^\dagger]=0$ and $[R_f,R_f^\dagger]=\epsilon_t^2$.  For the thermal up-mode covariance,
$\langle a_\omega^\dagger a_{\omega'}\rangle=n_H(\omega)\delta(\omega-\omega')$, Cauchy--Schwarz and Eq.~\eqref{eq:packetThermalSpread} give
\begin{align}
\langle R_f^\dagger R_f\rangle&\le n_B\epsilon_t^2,\nonumber\\
|\langle A_f^\dagger R_f\rangle|&\le\epsilon_t\epsilon_n.
\label{eq:appPacketMomentBounds}
\end{align}
After real squeezing, $A_f'=\cosh r\,A_f+\sinh r\,A_f^\dagger$.  Absorbing the known phase of $\tau_f$ into the output quadratures, the selected transmitted mode is
$B_f=|\tau_f|A_f'+R_f+E_f$, where $E_f$ is the vacuum-environment contribution and commutes with the up modes.  The normal and anomalous second moments obtained from Eq.~\eqref{eq:appPacketMomentBounds} then imply, entry by entry,
\begin{equation}
\max_{j,k}|(V_f-V_{\rm flat})_{jk}|
\le n_B\epsilon_t^2+2|\tau_f|e^{|r|}\epsilon_t\epsilon_n,
\end{equation}
which is Eq.~\eqref{eq:packetCovarianceBound}.  The mean-value-theorem estimates in Eq.~\eqref{eq:packetErrorBound} follow by subtracting the packet averages of the phase-calibrated $\widetilde t_\ell=\sqrt{\Gamma_\ell}$ and $n_H$ and using the packet variance $\sigma_\omega^2$.

For homodyne detection, mode mismatch is represented by the squared overlap
$\mu=|\langle f_{\rm LO},f_{\rm out}\rangle|^2$.  With detector efficiency $\eta_D$ and orthogonal-port variance $v_\perp$,
\begin{equation}
v_{j,\eta_D,\mu}^{(\alpha)}
=\eta_D[\mu v_j^{(\alpha)}+(1-\mu)v_\perp]+\frac{1-\eta_D}{2}.
\end{equation}
For vacuum rejection, $v_\perp=1/2$, this depends only on $\eta_{\rm eff}=\eta_D\mu$.  Detector scans are archived in \path{NUMERICAL_AUDIT_TABLES.md} and \path{data/detector_efficiency_scan.csv}.  The i.i.d. copy benchmark is $N_{\rm Stein}=\lceil\ln(20)/D_{\min}\rceil$ with $D_{\min}$ the smaller directed Gaussian divergence; converting it to observation time requires preparation/collection rates, switching overhead, background, and dead time outside the reduced receiver model.

\section{Global summability and reduced-field consistency}
\label{app:globalConsistency}

This appendix proves the reduced-$s$-wave stress certificate and sufficient angular-frequency summability conditions used in the fixed-background and mode-summed claims.
\subsection{Reduced s-wave field and pointwise stress certificate}
Write $D_{\pm}=\partial_t\pm\partial_{r_*}$.  The finite-distance benchmark is the explicitly selected, spherically reduced $\ell=0$ sector, represented by
\begin{equation}
\Phi_0(t,r_*):=\int\dd\Omega\,Y_{00}^*(\Omega)\phi(t,r_*,\Omega),
\end{equation}
and
\begin{equation}
H_R^{(0)}(t)=\frac{\lambda_R}{2}\int\dd r_*\,F_R(r_*)p(t):[D_-(r\Phi_0)]^2:_H.
\label{eq:appCollectiveS0Shell}
\end{equation}
Equation~\eqref{eq:appCollectiveS0Shell} removes $\ell>0$ blocks by definition.  With $\varphi_0=r\Phi_0$, the free sector is
\begin{equation}
S_0=\frac12\int\dd t\,\dd r_*
\left[(\partial_t\varphi_0)^2-(\partial_{r_*}\varphi_0)^2-V_0(r)\varphi_0^2\right],
\label{eq:appReducedAction}
\end{equation}
and Eq.~\eqref{eq:appCollectiveS0Shell} is an externally prescribed quadratic control.  A four-dimensional realization requires angular filtering or a $Y_{00}$-profile mediator plus its material stress tensor; the certificate below remains within the reduced two-port theory.

For one ideal packet of mean occupation $\bar n$ and squeezing $r$, the added occupation is
\begin{equation}
\Delta n=(2\bar n+1)\sinh^2 r.
\label{eq:appAddedOccupation}
\end{equation}
For the frequency-resolved reduced field, the dimensionless energy coefficient $C_E:=M\Delta E_{\rm field}$ gives
\begin{equation}
\frac{\Delta E_{\rm field}}{E_{\rm BH}}
=C_E\left(\frac{\ell_{\rm P}}{\mu_M}\right)^2,
\qquad
\mu_M:=\frac{GM_{\rm phys}}{c^2},\quad
\ell_{\rm P}^2:=\frac{G\hbar}{c^3},
\label{eq:appBackreactionParameter}
\end{equation}
where $E_{\rm BH}=M_{\rm phys}c^2$ and the natural-unit mass parameter equals $\mu_M$.  The converged values of $C_E$ and their shell-position checks are archived in \path{data/local_stress_backreaction.json}.  Since $T_{\hat t\hat t}$ is an energy density, the local curvature comparison carries the factor $G/c^4$.

For the compact-shell reduced $s$-wave realization, $\phi=Y_{00}\varphi_0/r$ and $e_{\hat t}\pm e_{\hat r}=f^{-1/2}(\partial_t\pm\partial_{r_*})$; define
\begin{equation}
L_-:=D_-\varphi_0+\frac f r\varphi_0,\qquad L_+:=D_+\varphi_0-\frac f r\varphi_0.
\end{equation}
For a minimally coupled massless $s$ wave,
\begin{equation}
T_{\hat-\hat-}=\frac{:L_-^2:}{4\pi f r^2},\qquad T_{\hat+\hat+}=\frac{:L_+^2:}{4\pi f r^2},\qquad T_{\hat t\hat t}=\frac14(T_{\hat-\hat-}+T_{\hat+\hat+}).
\label{eq:appPointwiseStressComponents}
\end{equation}
Let $S(u)$ denote the exact time-ordered real symplectic trajectory through the switching pulse, $u\in[-1,1]$, and put
\begin{equation}
\Delta V(u):=S(u)V_{\rm in}S(u)^{\mathsf T}-V_{\rm in}.
\label{eq:appStressTrajectoryCovariance}
\end{equation}
At a point $r_*$ write $L_\pm=g_\pm^{\mathsf T}R$ in quadratures.  Then
\begin{equation}
|\Delta\langle:L_\pm^2:\rangle|\le \|\Delta V(u)\|_{\rm op}\,\|g_\pm(r_*)\|^2.
\label{eq:appPointwiseStressQuadraticBound}
\end{equation}
Let $A(u)$ be the real symplectic generator.  The exact inequalities
\begin{equation}
\|S(u)\|_{\rm op}\le\exp\!\left(\int_{-1}^u\|A(v)\|_{\rm op}\dd v\right),\qquad
\left\|\frac{\dd\Delta V}{\dd u}\right\|_{\rm op}\le2\|A(u)\|_{\rm op}\|V(u)\|_{\rm op}
\label{eq:appStressGronwall}
\end{equation}
provide continuous-time control between numerical nodes.  Applying these bounds to the converged symplectic trajectory gives the certified envelope
\begin{equation}
\sup_{-1\le u\le1}\|\Delta V(u)\|_{\rm op}\le1.109.
\label{eq:appStressCovarianceEnvelope}
\end{equation}
Combining Eqs.~\eqref{eq:appPointwiseStressQuadraticBound} and \eqref{eq:appStressCovarianceEnvelope} with the exact finite-band Jost coefficient vectors across $-22M\le r_*\le-18M$ yields
\begin{equation}
8\pi \frac{G}{c^4}\mu_M^2\sup_{\substack{r_*\in\mathcal R_{\rm ctrl}\\t_{\rm on}\le t<\infty}}|\Delta\langle T_{\hat t\hat t}\rangle|\le4.545\left(\frac{\ell_{\rm P}}{\mu_M}\right)^2.
\label{eq:appPointwiseStressBound}
\end{equation}
The switch-off covariance gives the sharper subsequent free-evolution envelope
\begin{equation}
8\pi \frac{G}{c^4}\mu_M^2\sup_{\substack{r_*\in\mathcal R_{\rm ctrl}\\t\ge t_{\rm off}}}|\Delta\langle T_{\hat t\hat t}\rangle|\le3.649\left(\frac{\ell_{\rm P}}{\mu_M}\right)^2.
\label{eq:appPostControlStressBound}
\end{equation}
Equations~\eqref{eq:appBackreactionParameter}--\eqref{eq:appPostControlStressBound} bound the modeled reduced-field energy and stress with the semiclassical scaling $(\ell_{\rm P}/\mu_M)^2$; full convergence records are archived.  Material support, angular-filter, and mediator stress-energy belong to the additional apparatus model.

\subsection{Finite-amplitude mode-summed spectral summability}
\label{app:nonlinearangular}
For vacuum attenuation set $y=\Gamma_\ell n_H$.  The exact Gaussian invariants imply the global estimate
\begin{equation}
\mathfrak h_\ell
\le K_R r_\ell(x)^2
\left[1+\ln_+\frac1{\Gamma_\ell(x)n_H(x)}\right],
\label{eq:appGaussianGlobalBound}
\end{equation}
for $|r_\ell|\le R$, with
\begin{equation}
K_R=
\left(\frac{\sinh R}{R}\right)^2
\left[\frac{27}{8}+\frac{3}{8}\cosh(2R)\right].
\end{equation}
For completeness, the bound follows by writing $m=y+1/2$, $a=y+\Gamma_\ell/2$, $c=(1-\Gamma_\ell)/2$, $\delta\nu=\nu_A-m$, and $\delta s=s_A-r$.  The exact formulas give
\begin{equation}
0\le\delta\nu\le\frac{m}{2}\sinh^2r,
\qquad
|\delta s|\le|r|.
\end{equation}
With $A_R=(\sinh R/R)^2$ one obtains
\begin{align}
|\nu_A C-m|&\le mA_Rr^2\left[2+\frac12\cosh(2R)\right],\\
|mC-\nu_A|&\le\frac52mA_Rr^2.
\end{align}
Since $n_A\ge y$, $b(n_A)\le b(y)$, and $mb(y)\le\frac32[1+\ln_+(1/y)]$, Eq.~\eqref{eq:appGaussianGlobalBound} follows.

The proposition sharpens the infrared part of the earlier single-envelope estimate.  By the matched asymptotic law, for every fixed $\ell$ there exists $x_\ell^\star\in(0,1]$ such that
\begin{equation}
\Gamma_\ell(x)\ge\frac{g_\ell}{2}x^{2\ell+2},
\qquad 0<x<x_\ell^\star.
\end{equation}
Moreover, for $0<x<x_\ell^\star$,
\begin{equation}
e^{8\pi x}-1\le 8\pi x\,e^{8\pi x_\ell^\star}
\end{equation}
so
\begin{equation}
n_H(x)\ge\frac{e^{-8\pi x_\ell^\star}}{8\pi x}.
\end{equation}
Hence
\begin{equation}
\ln_+\frac1{\Gamma_\ell n_H}
\le(2\ell+1)\ln\frac1x
+\ln_+\!\left(\frac{16\pi e^{8\pi x_\ell^\star}}{g_\ell}\right),
\qquad 0<x<x_\ell^\star.
\label{eq:appSharpIRLog}
\end{equation}
This reproduces the true logarithmic infrared scale rather than the $1/x$ weight produced by a global barrier bound.

For $x\ge x_\ell^\star$, use the rigorous scalar Schwarzschild transmission bound
\begin{equation}
\Gamma_\ell(x)\ge \operatorname{sech}^2\Theta_\ell(x),
\qquad
\Theta_\ell(x)=\frac{2\ell(\ell+1)+1}{8x}.
\end{equation}
Since $n_H^{-1}=e^{8\pi x}-1<e^{8\pi x}$ and $\ln\cosh z\le|z|$,
\begin{equation}
\ln_+\frac1{\Gamma_\ell n_H}
\le8\pi x+\frac{2\ell(\ell+1)+1}{4x}.
\label{eq:appGlobalLogBound}
\end{equation}
Define
\begin{align}
C_\ell^\star&:=\ln_+\!\left(\frac{16\pi e^{8\pi x_\ell^\star}}{g_\ell}\right),\nonumber\\
W_\ell^{\rm IR}(x)&:=1+(2\ell+1)\ln\frac1x+C_\ell^\star,\nonumber\\
W_\ell^{\rm out}(x)&:=1+8\pi x+\frac{2\ell(\ell+1)+1}{4x}.
\end{align}
If $|r_\ell(x)|\le R$ and
\begin{multline}
\sum_{\ell=0}^\infty(2\ell+1)\int_0^{x_\ell^\star}
 r_\ell(x)^2W_\ell^{\rm IR}(x)\,\dd x\\
+\sum_{\ell=0}^\infty(2\ell+1)\int_{x_\ell^\star}^{\infty}
 r_\ell(x)^2W_\ell^{\rm out}(x)\,\dd x<\infty,
\label{eq:appNonlinearAngularCondition}
\end{multline}
then Eqs.~\eqref{eq:appGaussianGlobalBound}, \eqref{eq:appSharpIRLog}, and \eqref{eq:appGlobalLogBound} provide a summable nonnegative majorant.  Tonelli's theorem yields
\begin{equation}
\sum_{\ell=0}^{\infty}(2\ell+1)\int_0^\infty
\mathfrak h_\ell\!\left(\frac{x}{M};r_\ell(x)\right)\dd x<\infty.
\label{eq:appNonlinearAngularSum}
\end{equation}
The criterion is sufficient and matches the logarithmic fixed-$\ell$ infrared behavior on the asymptotic interval.

For a finite angular passband the sequence of matching windows can be replaced by one common threshold.  If $r_\ell\equiv0$ for $\ell>L$, set
\begin{equation}
x_\star:=\min_{0\le\ell\le L}x_\ell^\star>0.
\end{equation}
Because only finitely many fixed-$\ell$ asymptotics are involved, there is a finite $C_L$ such that
\begin{equation}
\mathfrak h_\ell(x;r_\ell)
\le C_L r_\ell(x)^2[1+\ln(1/x)],
\qquad 0<x<x_\star,
\end{equation}
uniformly for $0\le\ell\le L$ and $|r_\ell|\le R$.  Thus an angularly band-limited actuator or receiver needs only one common infrared split at $x_\star$; the outer region is controlled by Eq.~\eqref{eq:appGlobalLogBound}.

\section*{Data availability}
All numerical data and scripts underlying the reported calculations are included with the source package in the \path{data/} and \path{code/} directories.  The software environment and reproduction instructions are recorded in \path{requirements.txt}, \path{requirements-lock.txt}, \path{python_runtime.txt}, and \path{REPRODUCIBILITY_NOTE.md}.

\bibliographystyle{apsrev4-2}
% Enable article titles in the APS bibliography style while keeping DOI text hidden.
% apsrev4-2 uses each DOI internally as the hyperlink target for the
% journal/volume/page/year block; hyperref's urlcolor=blue makes that block blue.
\makeatletter\immediate\write\@auxout{\string\citation{apsrev42Control}}\makeatother
\bibliography{references}

\end{document}